\documentclass[10pt,journal,onecolumn]{IEEEtran} 
\usepackage{xcolor}
\usepackage{amsmath} 
\usepackage{amssymb}
\usepackage{amsfonts}

\usepackage{amsthm}
\usepackage{graphicx,url} 
\usepackage{bm}
\usepackage{float}
\usepackage{makecell}
\let\labelindent\relax
\usepackage{enumitem}
\usepackage{subcaption}
\usepackage{wrapfig}
\usepackage{algorithm,algorithmicx} 
\usepackage{algpseudocode} 
\usepackage{multirow}
\usepackage{diagbox}

\usepackage{cite}
\usepackage{soul}  
\usepackage[bookmarks=true]{hyperref}
\hypersetup{
    colorlinks=true,
    pdfborderstyle={/S/U/W 1},
    linkcolor=blue,
    filecolor=magenta,      
    urlcolor= blue,
    linkbordercolor=red,
}
\newcommand\rurl[1]{%
  \href{http://#1}{\nolinkurl{#1}}%
}

\usepackage{comment}

\let\abs\relax
\newcommand{\abs}[1]{\left\lvert#1\right\rvert}
\newcommand{\norm}[1]{\left\lVert#1\right\rVert}

\newcommand{\infnorm}[1]{\left\lVert#1\right\rVert_\infty}
\newcommand{\linfnorm}[1]{\left\lVert#1\right\rVert_{\mathcal{L}_\infty}}

\newcommand{\linfnormtruc}[2]{\left\lVert#1\right\rVert_{\mathcal{L}_\infty^{[0,#2]}}}

\newcommand{\lonenorm}[1]{\left\lVert#1\right\rVert_{\mathcal{L}_1}}

\def\lone{{\mathcal{L}_1}}

\def\lonew{${\mathcal{L}_1}$ }

\def\laplace#1{\mathfrak{L}\left[#1\right]}

\def \loneAC {$\lone$AC}

\def\nt {\textup{n}}

\def\tilx{\tilde{x}}

\def \hsigma{\hat{\sigma}}
\def\xin{x_\textup{in}}

\def\xr{x_\textup{r}}
\def \ur{u_\textup{r}}

\def \rt {\textup{r}}
\def \rhoin {\rho_\textup{in}}
\def \Zn {\mathbb{Z}_1^n}
\def \Zm {\mathbb{Z}_1^m}

\def \tint {\textup{int}}

\def \xn {x_\textup{n}}
\def \yn {y_\textup{n}}
\def \un {u_\textup{n}}

\def\mbR{\mathbb{R}}

\def\mbZ{\mathbb{Z}}

\def\mbZ{\mathbb{Z}}

\def\hsigma{\hat{\sigma}}

\def\mcC{\mathcal{C}}

\def\mcX{\mathcal{X}}
\def\mcH{\mathcal{H}}
\def\mcU{\mathcal{U}}
\def\mcG{\mathcal{G}}

\def\mrx{\mathrm{x}}
\def\mru{\mathrm{u}}
\def\mry{\mathrm{y}}

\def\trieq{\triangleq}

\newtheorem{theorem}{Theorem}
\newtheorem{lemma}{Lemma}

\theoremstyle{definition}  \newtheorem{definition}{Definition}
\theoremstyle{definition} \newtheorem{assumption}{Assumption}
\theoremstyle{remark}  
\newtheorem{remark}{Remark}

\usepackage{accents}

\makeatletter
\def\cl@part {\@elt {chapter}}
\makeatother

\usepackage{cleveref}

\crefname{equation}{}{} 
\crefname{lemma}{Lemma}{Lemmas}
\crefname{theorem}{Theorem}{Theorems}
\crefname{table}{Table}{Tables}
\crefname{figure}{Fig.}{Figs.}
\crefname{remark}{Remark}{Remarks}
\crefname{assumption}{Assumption}{Assumptions}
\crefname{section}{Section}{Sections}
\crefname{definition}{Definition}{Definitions}
\crefname{algorithm}{Algorithm}{Algorithms}

\makeatletter
\renewcommand*\env@matrix[1][\arraystretch]{%
  \edef\arraystretch{#1}%
  \hskip -\arraycolsep
  \let\@ifnextchar\new@ifnextchar
  \array{*\c@MaxMatrixCols c}}
\makeatother

\def\mbZ{\mathbb{Z}}

\def\mcC{{\mathcal{C}}}
\def\mcY{{\mathcal{Y}}}
\def\mcZ{{\mathcal{Z}}}

\def \mcU{{\mathcal{U}}}

\def \url {u_\textup{RL}}

\def \loneAC {$\lone$AC}

\def \loneRG {{$\lone$-RG}}

\usepackage{hyperref}
\usepackage{xcolor}
\hypersetup{
    colorlinks=true,
    pdfborderstyle={/S/U/W 1},
    linkcolor= blue,
    filecolor=magenta,      
    urlcolor= cyan,
    linkbordercolor=red,
    citecolor  = blue
}

\def \Oinf {O_\infty}
\def \tilOinf {{\tilde O}_\infty
}

\def \xcheckin  {{\check x}_\textup{in}}
\newcommand{\Tau}{\mathrm{T}}
\def \xcheckn {{\check x}_\textup{n}}
\newcommand{\interior}[1]{\textup{int}(#1)}
\def \mcV {\mathcal{V}}
\def \at {\textup{a}}

\usepackage{setspace}
\usepackage{algorithm,algorithmicx} 
\usepackage{algpseudocode} 

\algdef{SE}[DOWHILE]{Do}{doWhile}{\algorithmicdo}[1]{\algorithmicwhile\ #1}%

\begin{document}

\title{
Integrated Adaptive Control and Reference Governors for Constrained Systems with State-Dependent Uncertainties}

\author{Pan Zhao$^{1,\ast}$,~\IEEEmembership{Member,~IEEE,}
 Ilya Kolmanovsky$^2$,~\IEEEmembership{Fellow,~IEEE}
Naira Hovakimyan$^1$,~\IEEEmembership{Fellow,~IEEE}
\thanks{This work is supported by AFOSR, NASA and NSF under the NRI grant \#1830639, CPS grant \#1932529, and AI Institute: Planning grant \#2020289.}
\thanks{$^1$P.~Zhao and N.~Hovakimyan are with the Department of Mechanical Science and Engineering, University of Illinois at Urbana-Champaign, Urbana, IL 61801, USA. Email: \texttt{\{panzhao2, nhovakim\}@illinois.edu}. Corresponding author: P.~Zhao.}
\thanks{$^2$I.~Kolmanovsky is with the Department of Aerospace Engineering, University of Michigan, Ann Arbor, MI 48109, USA. Email: \texttt{ilya@umich.edu}.}
} 
\maketitle
\begin{abstract}
This paper presents an adaptive reference governor (RG) framework for a linear system with matched nonlinear  uncertainties that can depend on both time and states, subject to both state and input constraints. The proposed framework leverages an \lonew adaptive controller (\loneAC) that estimates and compensates for the uncertainties, and provides guaranteed transient performance, in terms of uniform bounds on the error between actual states and inputs and those of a nominal
(i.e., uncertainty-free) system. The uniform performance bounds provided by the \loneAC~are used to tighten the pre-specified state and control constraints. A reference governor is then designed for the nominal system using the tightened constraints, and guarantees robust constraint satisfaction. Moreover,  the conservatism introduced by the constraint tightening can be systematically reduced by tuning some parameters within the \loneAC. Compared with existing solutions, the proposed adaptive RG framework can potentially yield less conservative results for constraint enforcement due to the removal of uncertainty propagation along a prediction horizon, and improved tracking performance due to the inherent uncertainty compensation mechanism. Simulation results for a flight control example  illustrate the efficacy of the proposed framework. 
\end{abstract}
\onehalfspacing
\section{Introduction}\label{sec:intro}
There has been a growing interest in developing control methods that can handle state and/or input constraints. Examples of such constraints include actuator magnitude and rate limits, bounds imposed on process variables to ensure safe and efficient system operation, and collision/obstacle avoidance requirements. There are several choices for a control practitioner when dealing with constraints. One choice is to adopt the model predictive control (MPC) framework \cite{camacho2013mpc-book,rawlings2020mpc-book}, in which the state and input constraints can be incorporated into the optimization problem for computing the control signals. Another route is to augment a well-designed nominal controller, that already achieves high performance for small signals, with constraint handling capability that protects the system against constraint violations in transients for large signals.  The second route is attractive to practitioners who are interested in preserving an existing/legacy controller or are concerned with the computational cost, tuning complexity, stability, robustness, certification issues, and/or other requirements satisfactorily addressed by the existing controller. The reference governor (RG) is an example of the second approach. As its name suggests, RG is an {\it add-on} scheme for enforcing pointwise-in-time state and control constraints by modifying the reference command to a well-designed closed-loop system. The RG acts like a pre-filter that, based on the current value of the {\it desired reference command} $r(t)$ and of the states (measured or estimated) $x(t)$, generates a {\it modified reference command} $v(t)$ which avoids constraint violations. Since its advent, variants of RGs have been proposed for both linear and nonlinear systems. 
See the survey paper \cite{garone2017reference-governor-survey} and references therein. While RG has been extensively studied for systems  for which exact dynamic models are available, the design of RG for uncertain systems, i.e., systems with  unknown parameters, state-dependent uncertainties, unmodelled dynamics and/or external disturbances, has been less addressed. 
\subsection{Related Work}\label{sec:sub-related-work}
\noindent {\bf Robust Approaches}: As mentioned in \cite{garone2017reference-governor-survey}, the RG can be straightforwardly modified to handle unmeasured set-bounded disturbances  by taking into account all possible realizations of the disturbances when determining the {\it maximal output admissible set}  \cite{kolmanovsky1998robustRG-disturbance-invariant-sets}. For uncertain systems, various robust or tube MPC schemes  have also been proposed  \cite{kerrigan2001robust-mpc,langson2004robust-mpc-tube,rakovic2005robust-mpc,mayne2006robust-mpc,mayne2011tube-mpc-nonlinear,kohler2020computationally-rmpc,lopez2019dynamic-tube-mpc} and summarized in \cite{kouvaritakis2015mpc-classical-book}, most of which consider parametric uncertainties and bounded disturbances with only a few exceptions (e.g.,  \cite{kohler2020computationally-rmpc,lopez2019dynamic-tube-mpc}) that consider state-dependent uncertainties. However, robust approaches often lead to conservative results when the disturbances are large.

\noindent {\bf Adaptive and uncertainty compensation based approaches} could potentially achieve less conservative results than robust approaches. Along these lines, various adaptive MPC strategies with performance guarantees have been proposed for systems with unknown parameters \cite{zhang2020adaptive-mpc-parametric,adetola2009adaptive-mpc-nonlinear,pereida2021robust-adaptive-mpc} and state-dependent uncertainties \cite{wang2017adaptiveMPC, bujarbaruah2020semi-adaptive-mpc}. 
In particular,  \cite{pereida2021robust-adaptive-mpc} uses an \lonew adaptive controller \cite{naira2010l1book} to compensate for matched parametric uncertainties so that the uncertain plant behaves close to a nominal model, and uses robust MPC to handle the error between the combined system, consisting of the uncertain plant and the adaptive controller, and the nominal model. 
To the best of our knowledge, all of the existing adaptive MPC solutions, including \cite{pereida2021robust-adaptive-mpc} 
involve propagation of uncertainties along a prediction horizon. Reference \cite{poloni2014disturbance-constraint-rg} merged a Lyapunov function based RG  with a disturbance cancelling controller based on an input observer to achieve non-conservative treatment of uncertainties. Unfortunately, a bound on the rate of change of the disturbance is needed for the design, which is often difficult to obtain when the disturbance is dependent on states. Additionally, input constraints were not considered in that work.\\
\noindent {\bf State-dependent uncertainties (SDUs):} 
If a system is affected by SDUs, and the states are limited to a compact set, it is always possible to bound the SDU with a worst-case value and to apply the robust approaches (e.g., robust or tube MPC \cite{kerrigan2001robust-mpc,langson2004robust-mpc-tube,rakovic2005robust-mpc}) developed for bounded disturbances. However, by accounting for the state dependence, one can improve performance and reduce conservatism, as demonstrated in robust MPC solutions in \cite{pin2009robust,lopez2019dynamic-tube-mpc}. Adaptive MPC solutions which account for SDUs have been proposed in  \cite{wang2017adaptiveMPC,bujarbaruah2020semi-adaptive-mpc}. These solutions essentially rely on computing the uncertainty or state bounds along the prediction horizon using the Lipschitz proprieties of SDUs, and solving a robust MPC problem, using the computed bounds.  



\subsection{Contributions}
The contributions of this paper are as follows. Firstly, for constrained control under uncertainties, we develop an \loneRG~framework for linear systems with matched nonlinear uncertainties that could depend on both time and states, and with both input and state constraints. Our adaptive robust RG framework~leverages an \lonew adaptive controller (\loneAC) to estimate and compensate for the uncertainties, and to guarantee {\it uniform bounds} on the error between actual states and inputs and those of a nominal (i.e., uncertainty-free) closed-loop system. These uniform bounds characterize tubes in which actual states and control inputs are guaranteed to stay despite the uncertainties. A reference governor designed for the {nominal} system with { constraints} tightened using these uniform bounds guarantees robust constraint satisfaction in the presence of uncertainties.  
Additionally, we show that these {\it uniform bounds on state and input errors}, and thus the {\it conservatism induced by constraint tightening} 
can be {\it arbitrarily reduced} in theory by tuning the filter bandwidth and estimation sample time parameters of the \loneAC. Secondly, as a separate contribution to \lonew adaptive control, we propose a novel scaling technique that allows deriving {\it separate} tight uniform bounds on {\it each} state and adaptive control input, as opposed to a {\it single} bound for all states, or adaptive control inputs in existing \loneAC~solutions \cite{naira2010l1book}. The ability to provide such separate tight bounds makes an \loneAC~particularly attractive to be integrated with an RG for simultaneous constraint enforcement and improved trajectory tracking. Thirdly, we validate the efficacy of the proposed \loneRG~framework on a flight control example and we compare it with both baseline and robust RG solutions in simulations. 

Compared to existing literature, in particular, robust/adaptive MPC, \loneRG~has the following novel aspects:
\begin{itemize}[topsep=0pt]
    \item Thanks to the  uncertainty compensation and transient performance guarantees available for the \loneAC, \loneRG, (under suitable assumptions,) does {\it not require uncertainty propagation} along the prediction horizon. This uncertainty propagation is generally required in all existing robust and adaptive MPC approaches, 
    and incurs conservatism, which is avoided by \loneRG.  
    \item \loneRG~{\it simultaneously improves tracking performance and enforces the constraints}, while existing robust/disturbance-observer-based RG or robust/adaptive MPC solutions except a few such as \cite{pereida2021robust-adaptive-mpc,poloni2014disturbance-constraint-rg}, focus on constraint satisfaction only.   
    \item 
    Within \loneRG, the uniform bounds on the state and input errors (used for constraint tightening) and thus the conservatism induced by constraint tightening can be made arbitrarily small, which cannot be achieved by existing methods. 
    \item \loneRG~is able to handle  uncertainties that can nonlinearly depend on both time and states. Such a case has not been considered by previous adaptive MPC solutions that are based on uncertainty compensation. For instance, the solution in  \cite{pereida2021robust-adaptive-mpc}, which also leverages an \loneAC, only treats parametric uncertainties and state constraints.
\end{itemize}
The paper is structured as follows. \cref{sec:prob-statement} formally states the problem. \cref{sec:overview-preliminary} provides an overview of the proposed solution and discusses preliminaries related to RG and \loneAC~design. \cref{sec:l1ac-separate-bnds} introduces a scaling technique to derive separate and tight performance bounds for an \loneAC, while  
\cref{sec:l1rg} presents synthesis and performance analysis of the proposed \loneRG~framework. \cref{sec:sim} includes validation of the proposed \loneRG~framework on a flight control problem in simulations. 

{\it Notations}: Let $\mathbb{R}$, $\mathbb{R}_+$ and $\mbZ_+$  denote the set of real, non-negative real, and non-negative integer numbers, respectively.  $\mathbb{R}^n$ and  $\mathbb{R}^{m\times n}$  denote the $n$-dimensional real vector space and the set of real $m$ by $n$ matrices, respectively. $\mbZ_i$ and $\mbZ_1^n$ denote the integer sets $\{i, i+1, \cdots\}$ and $\{1, 2,\cdots,n\}$, respectively.
$I_n$  denotes an  identity matrix of size $n$, and $0$ is a zero matrix of a compatible dimension.
$\norm{\cdot}$ and  $\norm{\cdot}_\infty$ denote the $2$-norm and $\infty$-norm of a vector or a matrix, respectively. 
 The $\mathcal{L}_\infty$- and truncated $\mathcal{L}_\infty$-norm of a function $x:\mathbb{R}_+ \rightarrow\mathbb{R}^n$ are defined as $\norm{x}_{\mathcal{L}_\infty}\triangleq \sup_{t\geq 0}\infnorm{x(t)}$ and $\linfnormtruc{x}{T}\triangleq \sup_{0\leq t\leq T}\infnorm{x(t)}$, respectively. The Laplace transform of a function $x(t)$ is denoted by $x(s)\triangleq\mathfrak{L}[x(t)]$.
 For a vector $x$, $x_i$ denotes the $i$th element of $x$. Given a positive scalar $
 \rho$, $\Omega(\rho)\trieq \{z\in \mbR^n: \infnorm{z}\leq \rho \}$ denotes a high dimensional ball set of radius $\rho$ and  centered at the origin, while its dimension $n$ can be deduced from the context. For a high-dimensional set $\mcX$, $\textup{int}(\mcX)$ denotes the interior of $\mcX$ and $\mcX_i$ denotes the projection of $\mcX$ onto the $i$th coordinate. For given sets $\mcX,\mcY\subset \mbR^n$, $\mcX \oplus \mcY \trieq \{ x+y: x\in \mcX, y\in \mcY \}$ is the Minkowski set sum and $\mcX\ominus \mcY \trieq \{z: z+y\in \mcX, \forall y\in \mcY \}$ is the Pontryagin set difference. 

\section{Problem statement}\label{sec:prob-statement}
Consider an uncertain linear system represented by
\begin{equation}\label{eq:dynamics-uncertain-original}
\left\{ \begin{aligned}
  \dot x(t) &= Ax(t) + B(u(t) +  f(t,x(t))), \hfill \\
  y(t) &= Cx(t), \ x(0) = x_0,\\ 
\end{aligned}\right.    
\end{equation}
where $x(t)\in\mbR^n$, $u(t)\in \mbR^m$ and $y(t)\in\mbR^m$ are  the state, input and output vectors, respectively, $x_0\in\mbR^n$ is the initial state vector,  $f(t,x(t))\in \mbR^m$ denotes the uncertainty that can depend on both time and states, and $A,~B,$ and $C$ are matrices of compatible dimensions.
We want to design a control law for $u(t)$ such that the output vector $y(t)$ tracks a reference signal $r(t)$ while satisfying the specified state and control constraints:
\begin{equation}\label{eq:constraints}
    \begin{gathered}
  x(t) \in \mathcal{X},\quad u(t) \in \mathcal{U}, \quad \forall t\geq 0,\hfill
\end{gathered} 
\end{equation}
where $\mcX\subset \mbR^n$ and $\mcU \subset\mbR^m$ are pre-specified convex and compact sets with $0$ in the interior. Note that \cref{eq:constraints} can also represent constraints on some of the states and/or inputs. 

Suppose a baseline controller is available and  achieves desired performance for the nominal (i.e., uncertainty-free) system given a small desired reference command $r(t)$ to track. 
To enforce state and input constraints  \cref{eq:constraints} for the nominal system with larger signals,  one can simply leverage the conventional RG, which will generate a modified reference command $v(t)$ based on  $r(t)$. In such a case, the baseline controller can be selected as
\begin{equation}\label{eq:baseline-control-law}
   u_\textup{b}(t) = K_xx(t) + K_vv(t), 
\end{equation}
where $K_x$ and $K_v$ are feedback and feedforward gains. For both improved tracking performance and constraint enforcement in the presence of the uncertainty $f(t,x)$, we leverage an \loneAC. To this end, we adopt a compositional control law:
\begin{equation}\label{eq:total-control-law}
    u(t)=u_\textup{b}(t) +u_\at(t),
\end{equation}
where $u_\at(t)$ is the vector of the adaptive control inputs designed to cancel $f(t,x)$. With \cref{eq:baseline-control-law}, the uncertain system \cref{eq:dynamics-uncertain-original} can be rewritten as
\begin{equation}\label{eq:dynamics-uncertain}
\left\{
\begin{aligned}
  \dot x(t) &= {A_m}x(t) + {B_v}v(t) + B(u_\at(t) + f(t,x(t))), \\ 
  y(t) & = Cx(t), \ x(0) =x_0, 
\end{aligned}\right. 
\end{equation}
where ${A_m} \trieq A + B{K_x}$ is a Hurwitz matrix and ${B_v} \trieq B{K_v}$. 

The problem to be tackled can be stated as follows: {\it Given an uncertain system \cref{eq:dynamics-uncertain-original}, a baseline controller \cref{eq:baseline-control-law} and a desired reference signal $r(t)$,  design a RG (for determining $v(t)$) and the \loneAC~for $u_\at(t)$
 such that the output signal $y(t)$ tracks $r(t)$ whenever possible, 
while the state and input constraints \cref{eq:constraints} are satisfied.} We make the following assumption on the uncertainty. 
\begin{assumption}\label{assump:lipschitz-bnd-fi}
Given a compact set $\mcZ$, there exist known positive constants $L_{f_j,\mcZ}$, $l_{f_j,\mcZ}$ and $b_{f_j,\mcZ}$ ($j\in\Zm$) such that for any $x,z \in \mcZ$ and $t,\tau\geq 0$, the following inequalities hold for each $j\in\Zm$:
\begin{subequations}\label{eq:lipschitz-cond-and-bnd-fi}
\begin{align}
\abs{f_j(t,x) - f_j(\tau,z)}  & \le L_{f_j,\mcZ}\infnorm{x - z} + l_{f_j,\mcZ}\abs{t-\tau} , \label{eq:lipschitz-cond-fi}\\
  \abs {f_j(t,x)}  &\le b_{f_j,\mcZ}, \label{eq:bnd-fi} 
\end{align}
\end{subequations}
where $f_j(t,x)$ denotes the $i$th element of $f(t,x)$.
\end{assumption}

\begin{remark}
 \cref{assump:lipschitz-bnd-fi} indicates that in the compact set $\mcZ$,  $f_j(t,x)$ is {Lipschitz} continuous with respect to $x$  with a known Lipschitz constant $L_{f_j,\mcZ}$, has a bounded rate of variation $l_{f_j,\mcZ}$ with respect to $t$,  and is uniformly bounded by a constant $b_{f_j,\mcZ}$. 
\end{remark}
In fact, given the local Lipschitz constant $L_{f_j,\mcZ}$ and the bounded rate of variation $l_{f_j,\mcZ}$, a uniform bound for $f_j(t,x)$ in $\mcZ$ can always be derived if the bound on $f_j(t,x^\ast)$ for an arbitrary $x^\ast$ in $\mcZ$ and any $t\geq 0$ is known. For instance, assuming we know  $\abs{f_j(t,0)}\leq b^i_0$,  from \cref{eq:lipschitz-cond-fi}, we have that $\abs{f_j(t,x)-f_j(t,0)}\leq L_{f_j,\mcZ} \infnorm{x}$, which immediately leads to
    $\abs{f_j(t,x)}\leq b_0^i + L_{f_j,\mcZ}\max_{x\in\mcX} \infnorm{x}$,
for any $x\in \mcZ$ and $t\geq 0$. In practice, some prior knowledge about the uncertainty (e.g., $f_j$ depends on only a few instead of all states)  may be leveraged to obtain a tighter bound than the preceding one, 
derived using the Lipschitz continuity and triangular inequalities. This motivates the assumption on the uniform bound in \cref{eq:bnd-fi}. 

Under the conditions in \cref{assump:lipschitz-bnd-fi}, we immediately obtain that for any $x,z \in \mcZ$ and $t,\tau\geq 0$,
\begin{subequations}\label{eq:lipschitz-cond-bnd-f}
\begin{align}
\infnorm{f(t,x) - f(\tau,z)}  & \le{L_{f,\mcZ}}\infnorm{x - z} + l_{f,\mcZ}\abs{t-\tau} , \label{eq:lipschitz-cond-f}\\
  \infnorm {f(t,x)}  &\le {b_{f,\mcZ}},  \label{eq:bnd-f}
\end{align}
\end{subequations}
where 
\begin{equation}\label{eq:Lf-lf-bf-defn}
        L_{f,\mcZ} = \max_{j\in \Zm} L_{f_j,\mcZ}, \quad  l_{f,\mcZ} = \max_{j\in \Zm} l_{f_j,\mcZ}, \quad b_{f,\mcZ} = \max_{j\in \Zm} b_{f_j,\mcZ}.
\end{equation}
\begin{remark}
Our choice of making assumptions on $f_j(t,x)$ instead of on $f(t,x)$ as in \cref{eq:Lf-lf-bf-defn} facilitates deriving an  {\it individual} bound on each state 
and on each adaptive input (see \cref{sec:l1ac-separate-bnds} for details). 
\end{remark}

\begin{remark}
In principle, given the uniform bound on $f(t,x)$ in \cref{eq:bnd-f} obtained from \cref{assump:lipschitz-bnd-fi},  constraints can be enforced via robust RG  or robust MPC approaches that handle bounded disturbances, as discussed in \cref{sec:sub-related-work}.  However, when this bound is large, robust approaches can yield overly conservative performance. 
\end{remark}
 
 \section{Overview and Preliminaries}\label{sec:overview-preliminary}
 In this section, we first present an overview of the proposed \loneRG~framework and then introduce some preliminary results that provides a foundation for the \loneRG~framework. 
 
 \subsection{Overview of the \loneRG~Framework}\label{sec:sub-overview-l1rg}
Figure~\ref{fig:l1rg} depicts the proposed \loneRG~framework. 
As shown in \cref{fig:l1rg}, \loneRG~is comprised of two integrated components. The first one is an \loneAC~designed to compensate for the uncertainty $f(t,x)$ and to guarantee {\it uniform bounds} on the errors between actual states and inputs,  and those
of the nominal closed-loop system:
\begin{subequations}\label{eq:nominal-cl-system-w-un}
\begin{align}
  \dot x_\nt(t) &= {A_m}\xn(t) + {B_v}v(t), \ \xn(0) =x_0,\label{eq:nominal-cl-system} \\
     u_\nt(t)& =K_x \xn(t) + K_v v(t).\label{eq:un-defn}
\end{align}
\end{subequations}
 The second component is a RG designed for the {\it nominal system} \cref{eq:nominal-cl-system} with {\it tightened constraints} computed using the uniform bounds guaranteed by the \loneAC. 
  \begin{figure}[h]
     \centering
     \includegraphics[width = 0.6\columnwidth]{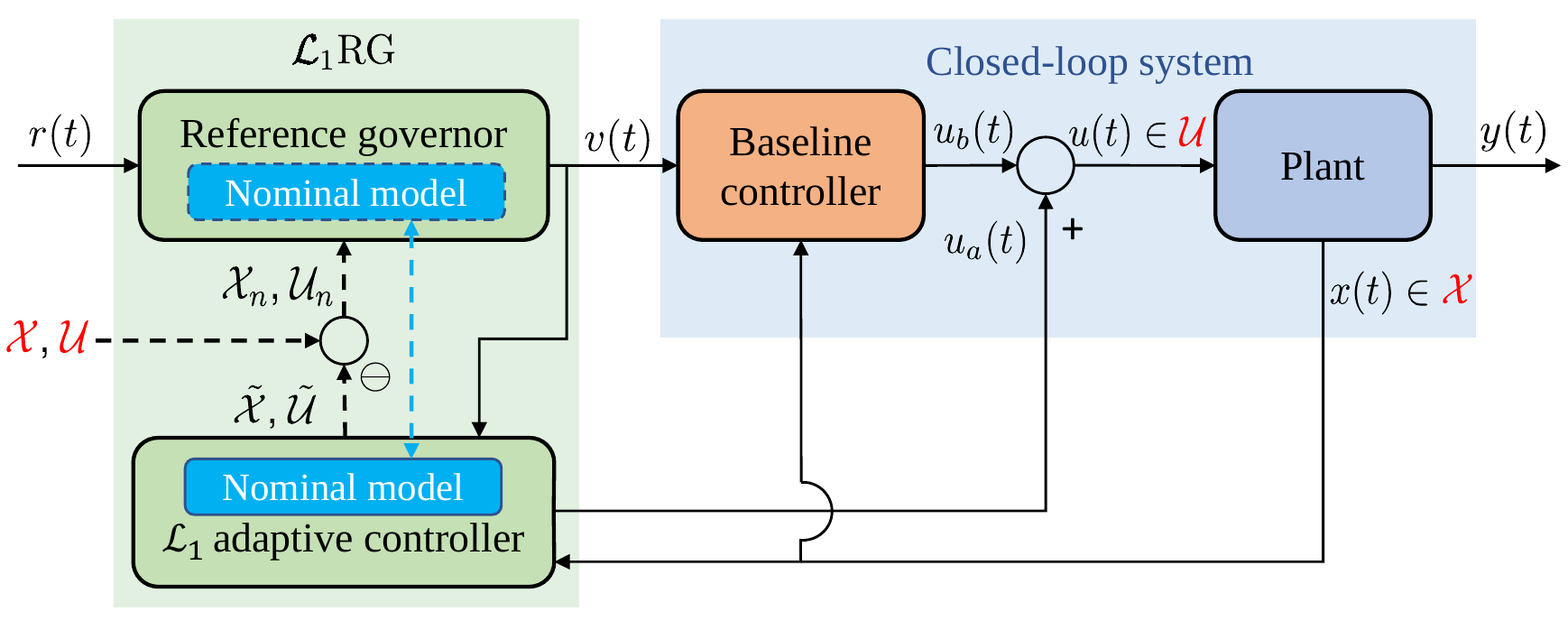}
     \caption{Diagram of the proposed \loneRG~framework}
     \label{fig:l1rg}
     \vspace{-3mm}
 \end{figure}
 More formally, we will design the \loneAC~to ensure 
\begin{equation}\label{eq:x-xn-in-X-u-un-in-U}
    \begin{aligned}
    x(t)-\xn(t) \in \tilde \mcX,  
  \quad  u(t)- \un(t) \in  \tilde \mcU, \quad \forall t\geq 0,
    \end{aligned}
\end{equation} 
where $x(t)$ and $u(t)$ are the vectors of states and of the total control inputs of the closed-loop system \cref{eq:dynamics-uncertain}: 
\begin{align}\label{eq:u-un-defn}
    u(t) - u_\nt(t) = K_x(x(t)-\xn(t)) + u_\at(t),
\end{align}
where $u(t)$ is given by \cref{eq:total-control-law} and $\un(t)$ by \cref{eq:un-defn},
and $\tilde \mcX$ and $\tilde \mcU$ are some pre-computed hyperrectangular sets  dependent on the properties of $f(t,x)$ and of the \loneAC. The details will be given in \cref{them:xi-uai-bnd} in  \cref{sec:sub-l1AC}. 
Define 
\begin{equation}\label{eq:Xn-Un-defn}
    \mcX_\nt \trieq \mcX \ominus \tilde \mcX, \quad \mcU _\nt \trieq \mcU \ominus \tilde \mcU. 
\end{equation}
Then for robust constraint enforcement, one just needs to design a RG for the nominal system \cref{eq:nominal-cl-system-w-un}
 with tightened constraints given by
\begin{equation}\label{eq:nominal-system-csts}
\xn(t) \in \mcX_\nt,\ u_\nt(t)\in \mcU _\nt, \quad \forall t\geq 0.
\end{equation}

\subsection{Reference Governor Design for a Nominal System}\label{sec:sub-rg-nominal}
 We now introduce the RG for the nominal system \cref{eq:nominal-cl-system-w-un} to enforce the constraints \cref{eq:nominal-system-csts}. We use the discrete-time RG approach of  \cite{garone2017reference-governor-survey} that uses a discrete-time model: 
  \begin{equation}\label{eq:nominal-system-discrete-w-un}
 \left\{
 \begin{aligned}
     \mathrm\xn(k+1)  &= \hat A_m \mathrm\xn(k) + \hat B_v  \mathrm v(k), \  \mathrm x(0) =x_0, \\
    \mathrm u_\nt(k) & =K_x \mathrm\xn(k) + K_v \mathrm v(k),
 \end{aligned}\right.
\end{equation}
where $\mathrm\xn(k)$, $\mathrm v(k)$ and $\mathrm \un(k)$ denotes the vectors of states, of reference command inputs, and of nominal control inputs, respectively, and $\hat A_m$ and $\hat B_v$ are computed from $A_m$ and $B_v$ in \cref{eq:nominal-cl-system-w-un} assuming a sampling time,  $T_d$.  When doing the discretization, we ensure that the discrete-time system \cref{eq:nominal-system-discrete-w-un} has the same states as the continuous-time system at all sampling instants. This can be achieved by using the zero-order hold discretization, since 
\begin{equation}\label{eq:v-t-kTd-relation}
    v(t)=\mathrm v (kT_d), \quad \forall t\in[kT_d,(k+1)T_d),
\end{equation}
which indicates that $v(t)$ is piecewise constant. 
 The constraints \cref{eq:nominal-system-csts} are imposed in discrete-time as 
 \begin{equation}\label{eq:nominal-system-discrete-csts}
  \mathrm \xn(k) \in \hat \mcX_\nt, \ 
   \mathrm\un(k)\in \hat \mcU _\nt, \quad \forall k\in\mbZ_+,
\end{equation}
where $\hat \mcX_\nt$ and $\hat \mcU_\nt$ 
are tightened versions of $\mcX_\nt $ and $\mcU_\nt $, respectively, introduced to avoid inter-sample constraint violations, and are defined by 
\begin{subequations}\label{eq:Xn-hat-Un-hat-defn}
\begin{align}
   \hat \mcX_\nt & \trieq \mcX_\nt \ominus \left\{ z\in\mbR^n: \infnorm{z}\leq \nu(T_d) \right\}, \\
   \hat \mcU_\nt & \trieq  \mcU_n \ominus \left\{ z\in\mbR^m: \infnorm{z}\leq 
\infnorm{K_x}\nu(T_d)\right\},
\end{align}
\end{subequations}
while  
\begin{equation}\label{eq:nu-Td-defn}
  \hspace{-2mm}  \nu(T_d) \!\trieq \!\max_{\tau\in[0,T_d]} \!\infnorm{e^{A_m \tau}\!\!-\!I_n}\!\max_{x\in \mcX_n,v\in \mathcal V}\!\infnorm{x\!+\!A_m^{-1}B_v v}\!, 
\end{equation}
with $\mathcal V$ denoting the set of all possible reference commands output by the RG.  

The following lemma formally guarantees that no inter-sample constraint violations will happen for the continuous-time system \cref{eq:nominal-cl-system-w-un} when the constraints for the discrete-time system \cref{eq:constraints-nom-ssform-discrete} are satisfied at all sampling instants. 
\begin{lemma}\label{lem:nominal-system-inter-sample-behavior}
Consider the continuous-time system \cref{eq:nominal-cl-system-w-un} and its discrete-time counterpart \cref{eq:nominal-system-discrete-w-un} that has the same states as \cref{eq:nominal-cl-system-w-un} at all sampling instants. If for the discrete-time system \cref{eq:nominal-system-discrete-w-un},
\begin{equation}\label{eq:xn-un-satisfy-tightened-csts}
   \mathrm  \xn(k)\in\hat \mcX_\nt,\  \mathrm \un(k)\in\hat \mcU_\nt,\quad k\in\mbZ_+,
\end{equation}
 with $\hat \mcX_\nt$ and $\hat \mcU_\nt$ defined in \cref{eq:Xn-hat-Un-hat-defn}, then \cref{eq:nominal-system-csts} holds for the continuous-time system \cref{eq:nominal-cl-system-w-un}. \end{lemma}
 \begin{proof}
 See \cref{sec:sub-proof-for-lem:nominal-system-inter-sample-behavior}.
 \end{proof}
 \begin{remark}\label{rem:Xn-hat-Xn-relation-w-Td}
 From \cref{eq:nu-Td-defn,eq:Xn-hat-Un-hat-defn}, we can see that $\hat \mcX_\nt$ and $\hat \mcU_\nt$ are close to $\mcX_\nt$ and $ \mcU_\nt$, respectively, when $T_d$ is small. 
 For practical implementation, inter-sample constraint violations may not be a big concern when $T_d$ is small. Under such case, we can simply set $\hat \mcX_\nt = \mcX_\nt$ and $\hat \mcU_\nt = \mcU_\nt$.  
 \end{remark}
Define
 \begin{equation}\label{eq:nominal-cl-system-output}
\hspace{-2mm} \mathrm y_{\nt}^c(k)  \! \trieq \!
  \begin{bmatrix}
  \mathrm\xn(k) \\
  K_x \mathrm \xn(k)\! + \!K_v \mathrm v(k) 
  \end{bmatrix} \!= \!\underbrace{\begin{bmatrix}
  I_n \\ K_x
  \end{bmatrix}}_{ \trieq\ C_c }
\mathrm\xn(k) \!+ \! \underbrace{\begin{bmatrix}
  0 \\ K_v
  \end{bmatrix}}_{ \trieq\ D_c } \mathrm v(k).
\end{equation}
Then, the constraints \cref{eq:nominal-system-discrete-csts} can be rewritten as 
  \begin{equation}\label{eq:constraints-nom-ssform-discrete}
    \mathrm y_{\nt}^c(k) \in  \mcY_\nt \trieq \hat \mcX_\nt \times \hat \mcU_\nt, \quad \forall k\in\mbZ_+,
\end{equation}
where  $\times$ denotes the cross product. 
\begin{remark}
In case there are no constraints on certain states and/or inputs, one can remove the rows of $C_c$ and $D_c$ defined in \cref{eq:nominal-cl-system-output} corresponding to these states and/or inputs,  and adjust the sets $\hat \mcX_\nt$, $\hat \mcU_\nt$ and $\mcY_\nt$ accordingly. 
\end{remark}

Similar to most RG schemes, the RG scheme we adopt here computes at each time instant a command $\mathrm v(k)$ such that, if it is {\it constantly} applied from the time instant $k$ onward, the ensuing output will always satisfy the constraints. More formally, we define the {\it maximal output admissible set} $O_\infty$ \cite{gilbert1991linear-constraints-maximal} as the set of all states $\mathrm \xn$ and inputs $\mathrm v$, such that the predicted response from the initial state $\mathrm\xn$ and with a constant input $\mathrm v$ satisfies the constraints \cref{eq:constraints-nom-ssform-discrete}, i.e.,
\begin{equation}\label{eq:Oinf-defn}
    O_\infty \trieq \{ (\mathrm v,\mathrm \xn): \hat {\mathrm y}_{\nt}^c(k|\mathrm v,\mathrm\xn) \in \mcY_n, \ \forall k\in \mbZ_+  \}, 
\end{equation}
where for system \cref{eq:nominal-system-discrete-w-un} the output prediction $\hat{\mathrm y}_{\nt}^c(k|\mathrm\xn,\mathrm v)$ is given by
\begin{align}
\hspace{-2.5mm}    \hat{\mathrm y}_{\nt}^c(k|\mathrm v,\mathrm\xn)  \!= ~C_c \hat A_m^k \mathrm\xn\! + \!C_c \sum_{j=1}^k \hat A_m^{j-1} \hat B_v \mathrm v \!+ \!\!D_c \mathrm v =   C_c \hat A_m^k \mathrm\xn  \!+\! C_c (I_n \!-\! \!\hat A_m)^{\!-1}(I_n\!-\!\!\hat A_m^k)\hat B_v \mathrm v \!+\!\! D_c \mathrm v.
\end{align}
Define $\tilde O_\infty$ as a slightly tightened version of $O_\infty$ obtained by constraining the command $\mathrm v$ so that the associated steady-state output $\bar { \mathrm y}_{\nt}^c= (D_c+C_c(I_n\!-\!\hat A_m)^{-1} \hat B_v)\mathrm v$ satisfies constraints with a nonzero (typically small) margin $\epsilon>0$, i.e., 
\begin{equation}
    \tilOinf = \Oinf \cap O^\epsilon,
\end{equation}
where 
    $O^\epsilon \trieq \{(\mathrm v,\mathrm\xn): \bar { \mathrm y}_{\nt}^c\in (1-\epsilon) \mcY_\nt \}.$
Clearly, $\tilOinf$ can be made arbitrarily close to $\Oinf$ by decreasing $\epsilon$. 
Based on the currently available state $\mathrm\xn(k)$ at an instant $k$, the RG computes $\mathrm v(k)$ so that
\begin{equation}
    (\mathrm v(k),\mathrm\xn(k))\in \tilde O_\infty.
\end{equation}
It is proven in \cite{gilbert1991linear-constraints-maximal} that if $\hat A_m$ is Schur, $(\hat A_m, C_c)$ is observable,  $\mcY_\nt$ is compact with $0$ in the interior, and $\epsilon>0$ is sufficiently small,  then the set $\tilOinf$ is {\it finitely determined}, i.e., there exists a finite index $k^\star$ such that 
\begin{equation}\label{eq:tilOinf-defn}
\tilOinf \! = \tilde O_{k^\star} = \!    \{ (\mathrm v,\mathrm\xn)\!: \hat { \mathrm y}_{\nt}^c(k|\mathrm v,\mathrm\xn) \! \in\! \mcY_n, \ k=0,1,\dots,k^\star \}  \cap O^\epsilon.
\end{equation}
Moreover,  $\tilOinf$ is {\it positively invariant}, which means that if $(\mathrm v(k),\mathrm\xn(k))\in \tilOinf$ and $\mathrm v(k)$ is applied to the system at time $k$, then $(\mathrm v(k),\mathrm\xn(k+1))\in \tilOinf$. Furthermore, if $\mcY_\nt$ is convex, then $\tilOinf$ is also convex. 
\begin{remark}
The process of computing $k^\star$ involves computing sets $\tilde O_k$ for $k=1,2,\dots$, and checking the condition $\tilde O_k = \tilde O_{k+1}$; $k^\star$ is the minimum $k$ for which this condition holds. 
\end{remark}

The proposed \loneRG~framework can leverage most of existing RG schemes developed for uncertainty-free systems. As an illustration and demonstration in \cref{sec:sim}, we choose the scalar RG introduced in \cite{bemporad1995nonlinear-rg,gilbert1995discrete-rg}. The scalar RG computes at each time instant $k$ a command $\mathrm v(k)$ which is the best approximation of the desired set-point $\mathrm r(k)$ along the line segment connecting $\mathrm v(k-1)$ and $\mathrm r(k)$ that ensures $(\mathrm v(k),\mathrm\xn(k))\in \tilOinf$. More specifically, the scalar RG solves at each discrete time $k$, the following optimization problem:
\begin{subequations}\label{eq:opt-prob-rg}
\begin{align}
    \kappa(k) =  \max_{\kappa\in[0,1]} &\kappa \\
    \textup{s.t. } &  \mathrm v =  \mathrm v(k-1)+\kappa(\mathrm r(k)-\mathrm r(k-1)),\\
    & (\mathrm v,\mathrm\xn(k))\in \tilOinf,
\end{align}
\end{subequations}
where $\kappa(k)$ is a scalar adjustable bandwidth parameter and $\mathrm v(k) = \mathrm v(k-1) + \kappa(k)(\mathrm r(k)-\mathrm v(k-1))$ is the modified reference command to be applied to the system. If there is no danger of constraint violation, $\kappa(k) =1$ and $\mathrm v(k)=\mathrm r(k)$ so that the RG does not interfere with the desired operation of the system. If $\mathrm v(k)=\mathrm r(k)$ would cause a constraint violation, the value of $\kappa(k)$ is decreased by the RG. In the extreme case, $\kappa(k) =0$, $\mathrm v(k)=\mathrm v(k-1)$, which means that the RG momentarily isolates the system from further variations of the reference command for constraint enforcement. Due to the positive invariance of $\tilOinf$, $\mathrm v(k)=\mathrm v(k-1)$ always satisfies the constraints, which ensures {\it recursive feasibility} under the condition that at $t=0$ a command $\mathrm v(0)$  is known such that $ \left(\mathrm v(0),\mathrm\xn(0)\right)\in \tilOinf$. Response properties of the scalar RG, including conditions for the {\it finite-time convergence} of $\mathrm v(k)$ to $\mathrm r(k)$ are detailed in \cite{gilbert1995discrete-rg}. 

\subsection{\lonew Adaptive Control Design and Uniform Performance Bounds}\label{sec:sub-l1AC}
We now present an \loneAC~that guarantees the bounds in \eqref{eq:x-xn-in-X-u-un-in-U}, {\it without} considering the state and control constraints in \cref{eq:constraints}. 
We first recall some basic definitions and facts from control theory, and introduce some definitions and lemmas.


\begin{definition}\label{definition:MIMO-L1} \cite[Section~III.F]{Scherer97Multi}
For a stable proper MIMO system $\mathcal H(s)$ with input $\mru(t)\in \mbR^m$ and output $\mry(t)\in \mbR^p$, its \lonew norm is defined as
\begin{equation}\label{eq:l1-norm-defn}
\lonenorm{\mathcal{H}(s)} \trieq \sup_{\mrx(0)=0,\linfnorm{\mru}\leq 1} {\linfnorm{\mry}}.
\vspace{-2mm}
\end{equation}
\end{definition}The following lemma follows directly from Definition~\ref{definition:MIMO-L1}. 
\begin{lemma}\label{lem:L1-Linf-relation}
For a stable proper MIMO system $\mathcal H(s)$ with states $\mrx(t)\in \mbR^n $, inputs $\mru(t)\in \mbR^m$ and outputs $\mry(t)\in \mbR^p$, under zero initial states, i.e., $\mrx(0) =0$, we have 
$\linfnormtruc{\mry}{\tau}\leq \lonenorm{\mathcal H(s)}\linfnormtruc{\mru}{\tau}$, for any $\tau\geq 0$. Furthermore, for any matrix $\Tau\in \mbR^{q\times p}$,  we have $\linfnorm{\Tau \mathcal H(s)}\leq \infnorm{\Tau}\linfnorm{\mathcal H(s)}$. 
\end{lemma}
A unique feature of an $\mathcal{L}_1$AC is a low-pass filter $\mcC(s)$ (with DC gain $\mcC(0)=I_m$)
that decouples the estimation loop from the control loop, thereby allowing for arbitrarily fast adaptation without sacrificing the robustness \cite{naira2010l1book}. For simplicity, we can select $\mcC(s)$ to be a first-order transfer function matrix
\begin{equation}\label{eq:filter-defn}
\mcC(s) = \textup{diag}(\mcC_1(s), \dots, \mcC_m(s)), \ \mcC_j(s) \trieq \frac{k_f^j}{(s+ k_f^j)},\ j\in \Zm,
\end{equation}
where  $k^j_f$ ($j\in\Zm$) is the bandwidth of the filter for the $j$th input channel. 
We now introduce a few notations that will be used later:
\begin{subequations}\label{eq:Hxm-Hxv-Gxm-defn}
\begin{align}
\hspace{-2.5mm}    \mcH_{xm}(s) &\!\trieq\! (sI_n \!-\! A_m)^{-1}\! B,\      \mcH_{xv}(s) \!\trieq\! (sI_n \!-\! A_m)^{-1} \!B_v, \label{eq:Hxm-Hxv-defn}\\
\hspace{-2.5mm}      \mcG_{xm}(s) & \!\trieq\! \mcH_{xm}(s)(I_m-\mcC(s)),\ 
   \label{eq:Gxm-Gm-defn}
\end{align}
\end{subequations}
where $A_m, B_v$ correspond to system \cref{eq:nominal-cl-system-w-un} and $B$ to \cref{eq:dynamics-uncertain-original}.
Also, letting $\xin(t)$ be the state of the system 
    $\dot{x}_\textup{in}(t) = A_m \xin(t), \  \xin(0) = x_0,$
we have $\xin(s) \trieq (sI_n-A_m)^{-1} x_0$. Defining
    $\rho_\textup{in} \trieq \lonenorm{s(sI_n-A_m)^{-1}}\max_{x_0\in \mcX_0}\infnorm{x_0}$,
and further considering that $A_m$ is Hurwitz and $\mcX_0$ is compact, we have $\linfnorm{\xin}\leq \rhoin$ according to \cref{lem:L1-Linf-relation}. 
\subsubsection{\lonew adaptive control architecture}
For stability guarantees, the filter $\mcC(s)$ in \eqref{eq:filter-defn} needs to ensure that there exists  a positive constant $\rho_r$ and a (small) positive constant $\gamma_1$ such that
\begin{subequations}\label{eq:l1-stability-condition-w-extra-Lf}
\begin{align}
\lonenorm{\mcG_{xm}(s)}b_{f,\mcX_r}  < \rho_r - \lonenorm{\mcH_{xv}(s)}&\linfnorm{v} - \rho_\textup{in},\label{eq:l1-stability-condition} \\ 
\lonenorm{\mcG_{xm}(s)}L_{f,\mcX_a}&<1,\label{eq:l1-stability-condition-Lf}
\end{align}
\end{subequations}
where 
\begin{align} 
 \rho  & \trieq \rho_r + \gamma_1, \label{eq:rho-defn} \\ 
 \mcX_r & \trieq \Omega(\rho_r), \   \mcX_a  \trieq  \Omega(\rho). \label{eq:Xr-Xa-defn}
\end{align}
\begin{remark}
We will show in \cref{lem:ref-xr-ur-bnd} and \cref{them:x-xref-bnd} that $\rho_r$ and $\rho$ are actually uniform bounds on the states of a non-adaptive reference system (defined in \eqref{eq:ref-system}) and of the 
adaptive system, respectively. 
\end{remark}
\begin{remark}\label{rem:Gxm_linfnorm-stability-filter}
Note that $\lonenorm{\mcG_{xm}(s)}\rightarrow 0$, when the bandwidth of the filter $\mcC(s)$ goes to infinity, i.e.,  $k_f^j\rightarrow \infty$ for all $j\in \Zm$. Furthermore, $b_{f,\Omega(\rho_r)}$ can be bounded using the Lipschitz property \cref{eq:lipschitz-cond-f} of $f(t,x)$ in $\Omega(\rho_r)$, and $L_{f,\Omega(\rho)}$ is bounded given any $\rho>0$. Therefore, \cref{eq:l1-stability-condition-w-extra-Lf} can always be satisfied under  a sufficiently high bandwidth for $\mcC(s)$.
\end{remark}
 
A typical \loneAC~is  comprised of three elements, namely a state predictor, an adaptive law and a low-pass filtered control law. For the system \cref{eq:dynamics-uncertain}, the {\bf state predictor}  is defined by 
\begin{equation}\label{eq:state-predictor}
\dot {\hat x} (t)= A_m x(t) + B_v v(t) + B( u_\at (t)+ \hsigma_1(t)) + B^\perp \hsigma_2(t)  + A_e \tilx(t), \ \hat x(0) = x_0,
\end{equation}
where $\tilx(t) = \hat x(t) - x(t)$ is the prediction error, $A_e$ is a Hurwitz matrix, $B^\perp \in \mbR^{n\times (n-m)}$ is an arbitrary matrix satisfying $B^\perp B = 0$ and $\textup{rank}\left(\left[B\ B^\perp \right]\right)=n$, and  $\hsigma_1(t)$ and $\hsigma_2(t)$ are estimated matched and unmatched disturbances, respectively. The estimates $\hsigma_1(t)$ and $\hsigma_2(t)$ are updated by the following piecewise-constant {\bf adaptive law} (similar to that in \cite[Section~3.3]{naira2010l1book}):
\begin{equation}\label{eq:adaptive_law}\left\{
\begin{aligned}
&\begin{bmatrix}
      \hsigma_1(t) \\
      \hsigma_2(t)
\end{bmatrix}  &  &  \hspace{-4mm}=  
\begin{bmatrix}
      \hsigma_1(iT) \\
      \hsigma_2(iT)
\end{bmatrix}
    , \quad t\in [iT, (i+1)T), \\
& \begin{bmatrix}
      \hsigma_1(iT) \\
      \hsigma_2(iT)
\end{bmatrix} &
& \hspace{-4mm}= -\!\left[B\ B^\perp\right]^{-1} \Phi^{-1}(T)e^{A_eT}\tilde{x}(iT),
\end{aligned}\right.
\end{equation}
where $T$ is the estimation sampling time and $\Phi(T)\trieq A_e^{-1}\left(e^{A_eT}\!-I_n\right)$. Finally, the {\bf control law} is given by 

\begin{equation}\label{eq:l1-control-law}
    u_\at(s) = -\mcC(s)\laplace{\hsigma_1(t)}. 
\end{equation}
The control law \cref{eq:l1-control-law} tries to cancel the estimated (matched) uncertainty within the bandwidth of the filter $\mcC(s)$. Additionally, unmatched uncertainty estimate ($\hsigma_2(t)$) appears in \cref{eq:state-predictor,eq:adaptive_law}, although the system dynamics \cref{eq:dynamics-uncertain} contains only matched uncertainty. This is due to the adoption of the piecewise-constant adaptive law, which may produce nonzero value for $\hsigma_2(t)$. However, a non-zero $\hsigma_2(t)$ will not cause an issue either for implementation or for performance guarantee. Additionally, it is possible to prove that $\lim_{T\rightarrow 0} \hsigma_2(t) = 0 $ for any $t\geq 0$ \cite{zhao2021RALPV}, i.e., the estimated unmatched uncertainty will be close to zero when $T$ is  small.

\subsubsection{Uniform performance bounds}
We first define some constants:
\begin{subequations}\label{eq:alpha_012-gamma0-T-defn}
\begin{align}
    \bar \alpha_0(T) & \trieq \int_0^T\infnorm{e^{A_e(T-\tau)}B}d\tau,\label{eq:alpha_0-defn}\\
    \bar \alpha_1(T) &\trieq \max_{t\in [0,T]} \infnorm{e^{A_e t}},
   \label{eq:alpha_1-defn} \\
      \bar \alpha_2(T) &\trieq \max_{t\in [0,T]}\int_0^t \infnorm{e^{A_e (t-\tau)}\Phi^{-1}(T)e^{A_eT}}d\tau,\label{eq:alpha_2-defn}
        \\     \gamma_0(T) &\trieq b_{f,\mcX_a} \bar\alpha_0(T)\left(\bar a_1(T) + \bar a_2(T)+1\right), \label{eq:gamma_0-T-defn} 
\end{align}
\end{subequations}
where $\alpha_0(T),~\alpha_1(T)$ and $\alpha_2(T)$ are defined in \cref{eq:alpha_0-defn}, \cref{eq:alpha_1-defn,eq:alpha_2-defn}, respectively.
Clearly,  $b_{f,\mcX_a}$ for a compact set $\mcX_a$ and $\lim_{T\rightarrow 0}\bar \alpha_1(T) =0 $ are bounded, and $\lim_{T\rightarrow 0}\bar \alpha_0(T) =0$. By using Taylor series expansion of $e^{A_eT}$, one can show that $\lim_{T\rightarrow 0} \int_0^T \infnorm{\Phi^{-1}(T)}d\tau$ is bounded, which implies that $ \lim_{T\rightarrow 0} \bar \alpha_2(T)$ is bounded. As a result, we have
\begin{equation}\label{eq:gammaT-to-0}
    \lim_{T\rightarrow 0} \gamma_0(T) = 0.
\end{equation}
Further define 
\begin{align}
 \hspace{-2mm}    \rho_{ur} \trieq & \lonenorm{\mcC(s)} b_{f,\mcX_r}, \label{eq:rho_ur-defn} \\
\hspace{-2mm}  \gamma_2 \trieq  &  \lonenorm{\mcC(s)}\! L_{f,\mcX_a} \gamma_1 \!+\! \lonenorm{\mcC(s)B^\dagger(sI_n\!-\!A_e)}\!\gamma_0(T), \label{eq:gamma2-defn} \\
 \hspace{-2mm}  \rho_{u_\at} \trieq & \rho_{ur} +\gamma_2, \label{eq:rho-u-defn}
\end{align}
where $\gamma_1$ is introduced in \cref{eq:rho-defn}. Due to \cref{eq:gammaT-to-0,eq:l1-stability-condition-Lf}, we can always select a small enough $T>0$ such that 
\begin{equation}\label{eq:T-constraint}
\frac{\lonenorm{\mcH_{xm}(s)\mcC(s)B^\dagger(sI_n-A_e)}}{1-\lonenorm{\mcG_{xm}(s)}L_{f,\mcX_a}}\gamma_0(T) < \gamma_1,
\end{equation}
 where $\mcX_a$ is defined in \cref{eq:Xr-Xa-defn} and $B^\dagger$ is the pseudo-inverse of $B$. 

Following the convention for performance analysis of an \loneAC \cite{naira2010l1book}, we introduce the following reference system: 
\begin{subequations}\label{eq:ref-system}
\begin{align}
  \dot\xr(t) &\! =\! {A_m}x_\textup{r}(t) \!+\! {B_v}v(t) \!+\! B({u_{\textup{r}}}(t) \!+\! f(t,x_\textup{r}(t))), \hfill \\
  {u_{\textup{r}}}(s) & \! =\! -\mcC(s)\laplace{f(t,x_\textup{r}(t))}, \quad\xr(0)\! =\!x_0, \label{eq:ref-system-uar}
\end{align} 
\end{subequations}
Clearly, the control law in the reference system \cref{eq:ref-system} {\it partially} cancels the uncertainty $f(t,x_\textup{r}(t)))$ within the bandwidth of the filter $\mcC(s)$. Moreover, the control law depends on the true uncertainties and is thus {\it not implementable}. The reference system is introduced to help characterize the performance of the adaptive closed-loop system, which will be done in four sequential steps: (i) establishing the bounds on the states and inputs of the reference system (\cref{lem:ref-xr-ur-bnd}); (ii) quantifying the difference between the states and inputs of the adaptive system and those of the reference system (\cref{them:x-xref-bnd}); (iii) quantifying the difference between the states and inputs of the reference system and those of the nominal system (\cref{lem:ref-id-bnd}); (iv) based on the results from (ii) and (iii),  quantifying the difference between the states and inputs of the adaptive system and those of the nominal system (\cref{them:x-xid-bnd}).

The proofs of these lemmas and theorems mostly follow the typical \loneAC~analysis procedure \cite{naira2010l1book}, and are included in appendices for completeness. 

For notation brevity, we define:
\begin{equation}\label{eq:eta-eta_r-defn}
    \eta(t) \trieq f(t,x(t)),\quad \eta_\rt(t) \trieq f(t,x_\rt(t)).
\end{equation}
To provide an overview, \cref{table:diff-sys-in-l1AC-summary} summarizes the different (error) systems involved in this section and their related theorems/lemmas, the uniform bounds, the \loneAC~parameters and conditions.
{\setlength{\arraycolsep}{1pt}
\begin{table}[h]
\centering
\caption{An overview of different (error) systems involved in \cref{sec:sub-l1AC}, and their related theorem/lemma, uniform bounds, \loneAC~parameters and conditions}\label{table:diff-sys-in-l1AC-summary}
\footnotesize
\begin{tabular}{|l|l|l|l|l|l|}
\hline
& {\bf (Error) System}    & \makecell[l]{\bf Theorem/Lemma} & {\bf Uniform Bounds on States and Inputs} & \makecell[l]{\bf \loneAC~Parameters} & \makecell[l]{\bf Conditions} \\ 
\hline\hline
1 & Nominal system \cref{eq:nominal-cl-system-w-un}  & \cref{lem:L1-Linf-relation}  & $\linfnorm{\xn}\! \leq \! \rhoin + \lonenorm{\mcH_{xv}}\infnorm{v}$     &  N/A     &    N/A    \\ \hline
2 & Reference system \cref{eq:ref-system}  & \cref{lem:ref-xr-ur-bnd}     &  $\linfnorm{x_\textup{r}} \! < \! \rho_r$,   $\linfnorm{u_{\textup{r}}} \! <\! \rho_{ur}$  &  $\mcC(s)$  &   \cref{eq:l1-stability-condition}    \\ \hline
3 & \makecell[l]{Diff. b/t reference and adaptive systems}  & \cref{them:x-xref-bnd}  &   $ \linfnorm{\xr\!-\!x}  \! \leq \! \gamma_1,\
\linfnorm{u_{\textup{r}} \!-\! u_\at} \! \leq \! \gamma_2$ & $A_e,~T,~\mcC(s)$ & \cref{eq:T-constraint,eq:l1-stability-condition} 
\\  \hline

4 & \makecell[l]{Diff. b/t reference and nominal systems \\ }   &    \cref{lem:ref-id-bnd}                 &   \makecell[l]{$\left\| \xr \!-\!\xn \right\|_{{\mathcal L}{_\infty }} \! \leq \!  \lonenorm{\mcG_{xm}}b_{f,\mcX_r}$ 
}   &  $\mcC(s)$   &  \cref{eq:l1-stability-condition}  \\ \hline\hline

5 &\makecell[l]{\bf Adaptive system:  \cref{eq:dynamics-uncertain} and the \loneAC}  & \cref{them:x-xref-bnd}      &   $\linfnorm{x} <\rho$,   $\linfnorm{u_\at} <\rho_{u}$     &   $A_e,~T,~\mcC(s)$    &   \cref{eq:T-constraint,eq:l1-stability-condition}     \\ \hline

6 &\makecell[l]{\bf Diff. b/t adaptive  and nominal systems} & \cref{them:x-xid-bnd} &  $\linfnorm{x \!-\! \xn} \! \leq \! \tilde \rho$ 
& $A_e,~T,~\mcC(s)$ &  \cref{eq:T-constraint,eq:l1-stability-condition}    \\ \hline
\end{tabular}
\vspace{-3mm}
\end{table}}
The proofs for \cref{lem:ref-xr-ur-bnd,lem:xtilde-bnd,them:x-xref-bnd,lem:ref-id-bnd} are given in  \cref{sec:sub-proof-for-lem:ref-xr-ur-bnd,sec:sub-proof-lem:xtilde-bnd,sec:sub-proof-them:x-xref-bnd,sec:sub-proof-lem:ref-id-bound}.
\begin{lemma} \label{lem:ref-xr-ur-bnd}
For the closed-loop reference system in \eqref{eq:ref-system} subject to \cref{assump:lipschitz-bnd-fi} and the stability condition in  \eqref{eq:l1-stability-condition}, we have
\vspace{-4mm}\\
\begin{align}
 \linfnorm{x_\textup{r}} &<\rho_r \label{eq:xref-bnd}, \\ \linfnorm{u_{\textup{r}}} &<\rho_{ur}, \label{eq:uref-bnd}
\end{align}where $\rho_r$ is introduced in \cref{eq:l1-stability-condition}, and $\rho_{ur}$ is defined in \cref{eq:rho_ur-defn}.
\end{lemma}

From \cref{eq:dynamics-uncertain,eq:state-predictor}, the prediction error dynamics are given by 
\begin{equation}\label{eq:prediction-error}
\begin{aligned}
      \dot{\tilx}(t) &= A_e \tilx(t) + B\left(\hsigma_1(t) - f(t,x(t))\right) + B^\perp \hsigma_2(t).
\end{aligned}
\end{equation}
The following lemma establishes a bound on the prediction error under the assumption that the actual states and adaptive inputs are bounded. 
\begin{lemma}\label{lem:xtilde-bnd}
Given the uncertain system \eqref{eq:dynamics-uncertain} subject to \cref{assump:lipschitz-bnd-fi}, the state predictor  \eqref{eq:state-predictor} and the adaptive law \eqref{eq:adaptive_law}, if 
\begin{equation}\label{eq:x-u-tau-bnd-assump-in-lemma}
    \linfnormtruc{x}{\tau}\leq \rho, \quad \linfnormtruc{u_\at}{\tau}\leq \rho_{u_\at},
\end{equation}
with $\rho$ and $\rho_{u_\at}$ defined in \cref{eq:rho-u-defn,eq:rho-defn}, respectively, then 
\begin{align}
  \linfnormtruc{\tilx}{\tau} \leq \gamma_0(T). \label{eq:tilx_tau-leq-gamma0}
\end{align}
\end{lemma}
\begin{theorem}\label{them:x-xref-bnd}
{Given the uncertain system \eqref{eq:dynamics-uncertain} subject to  \cref{assump:lipschitz-bnd-fi} and  the reference system \eqref{eq:ref-system} subject to the conditions \cref{eq:l1-stability-condition-Lf,eq:l1-stability-condition}  with a constant $\gamma_1>0$, with the \loneAC~defined via \cref{eq:state-predictor,eq:adaptive_law,eq:l1-control-law} subject to the sample time constraint  \eqref{eq:T-constraint}, we have }
\begin{subequations}
\begin{align}
    \linfnorm{x} & \leq \rho, \label{eq:x-bnd}\\
    \linfnorm{u_\at} & \leq \rho_{u_\at}, \label{eq:ua-bnd}\\
    \linfnorm{x_\textup{r}-x} &\leq \gamma_1, \label{eq:xref-x-bnd}\\
    \linfnorm{u_{\textup{r}}-u_\at} &\leq \gamma_2, \label{eq:uref-u-bnd}
\end{align}
\end{subequations}
where $\rho$,$\gamma_2$ and $\rho_{u_\at}$ are defined in \cref{eq:rho-defn}, 
 \cref{eq:gamma2-defn,eq:rho-u-defn}, respectively. 
\end{theorem}
\begin{remark}\label{rem:ref-ad-bnd-T-discussion}
For an arbitrarily small $\gamma_1>0$, one can always find a small enough $T$ such that the constraint \cref{eq:T-constraint} is satisfied.  According to \cref{eq:gamma2-defn}, $\gamma_2$ depends on $\gamma_1$ and $\gamma_0(T)$, and can be made arbitrarily small by reducing $\gamma_1$ and $T$. Thus, by reducing $T$, both $\gamma_1$ and $\gamma_2$ can be made arbitrarily small, which indicates that the difference between the inputs and states of the adaptive system and those of the reference system can be made arbitrarily small from \cref{them:x-xref-bnd}.
\vspace{0mm} \end{remark}
\begin{lemma} \label{lem:ref-id-bnd}
Given the reference system \eqref{eq:ref-system} and the nominal system \eqref{eq:nominal-cl-system},  subject to \cref{assump:lipschitz-bnd-fi}, and the condition  \cref{eq:l1-stability-condition}, we have
\begin{align}
{\left\| {{x_\textup{r}} - \xn} \right\|_{{\mathcal L}{_\infty }}} & \leq  \lonenorm{\mcG_{xm}}b_{f,\mcX_r} \label{eq:xref-xid-bnd} 
\end{align}
\end{lemma}
\begin{remark}\label{rem:xref-xnom-diff-discussion}
When the bandwidth of the filter $\mcC(s)$ goes to infinity,  $\lonenorm{\mcG_{xm}}$ and thus  ${\left\| {{x_\textup{r}} - \xn} \right\|_{{\mathcal L}{_\infty }}}$ go to 0. This indicates that the difference between the states of the reference system and those of the nominal system can be made arbitrarily small by increasing the filter bandwidth.  However, a high-bandwidth filter allows for high-frequency control signals to enter the system under fast adaptation (corresponding to small $T$), compromising the robustness. 
Thus, the filter presents a trade-off between  robustness and performance.  More details about the role and design of the filter can be found in \cite{naira2010l1book}.
\end{remark}

From \cref{them:x-xref-bnd},  \cref{lem:ref-id-bnd} and application of the triangle inequality, we can obtain uniform bounds on the error between the actual system \cref{eq:dynamics-uncertain} and the nominal system \cref{eq:nominal-cl-system}, formally stated in the following theorem. The proof is straightforward and thus omitted.   
\begin{theorem}\label{them:x-xid-bnd}
Given the uncertain system \eqref{eq:dynamics-uncertain} subject to \cref{assump:lipschitz-bnd-fi}, the nominal system \eqref{eq:nominal-cl-system}, and the \loneAC~defined via \cref{eq:state-predictor,eq:adaptive_law,eq:l1-control-law} subject to the conditions \cref{eq:l1-stability-condition-Lf,eq:l1-stability-condition} with a constant $\gamma_1>0$ and the sample time constraint \eqref{eq:T-constraint}, we have
\begin{align}
\linfnorm{x-\xn} &\leq \tilde \rho, \label{eq:x-xid-bnd} \\
 \linfnorm{u_\at} &\leq \rho_{u_\at}, \label{eq:ua-bnd-2}
\end{align} where $\rho_{u_\at}$ is defined in \cref{eq:rho-u-defn}, and
\begin{align}
    \tilde \rho &\trieq \lonenorm{\mcG_{xm}(s)}b_{f,\mcX_r}+\gamma_1. \label{eq:tilrho-defn}
\end{align}
\end{theorem}
\begin{remark}\label{rem:x-xid-discussion-Cs-T}
From \cref{rem:xref-xnom-diff-discussion,rem:ref-ad-bnd-T-discussion}, by decreasing $T$ and increasing the bandwidth of the filter $\mcC(s)$, one can make (i) {the states of the adaptive system arbitrarily close to those of the nominal system}; and (ii) the adaptive inputs $u_\at(t)$ arbitrarily close to $f(t,x)$, i.e., the true uncertainty, since $f(t,\xr)$ is arbitrarily close to $f(t,x)$ when the error between $x(t)$ and $\xr(t)$ is arbitrarily small.
\end{remark}

\section{\loneAC~with Separate Bounds for States and Inputs}\label{sec:l1ac-separate-bnds}
 In \cref{sec:sub-l1AC}, we presented an \loneAC~that guarantees uniform bounds on the states and adaptive control inputs of the adaptive system with respect to the nominal system, without consideration of the constraints \cref{eq:constraints}. However, as can be seen from \cref{them:x-xid-bnd}, the uniform bound on $x(t)-\xn(t)$ or $u_\at(t)$ is represented by the vector-$\infty$ norm, which always leads to the {\it same bound} for all the states, $x_i-x_{\nt,i}(t)$ ($i\in\mbZ_1^n$), or all the adaptive inputs, $u_{\at,j}$ ($j\in\mbZ_1^m$). The use of vector-$\infty$ norms may lead to conservative bounds for some specific states or adaptive inputs, making it impossible to satisfy the constraints \cref{eq:constraints} or leading to significantly tightened constraints for the RG design. 
To reduce such conservatism, this section will present a scaling technique to  derive an {\it individual bound} for each $x_i(t)-x_{\textup n,i}(t)$ ($i\in\Zn$) and $u_{\at,j}(t)$ ($j\in\Zm$).

From \cref{them:x-xid-bnd}, one can see that the bound on $x(t)-\xn(t)$ (or $u_\at(t)$) consists of two parts: the first part is $\gamma_1$ (or $\gamma_2$) that can be made arbitrarily small by reducing $T$ (see \cref{rem:ref-ad-bnd-T-discussion}), while the second part is a bound on $\xr(t)-\xn(t)$ (or $\ur(t)$). 
Next, we will derive an individual bound for each $x_i(t)-x_{\textup r,t}(t)$ (or $u_{\rt,j}(t)$). \\
\textbf{Derive Separate Bounds for States via Scaling}:
For deriving an individual bound for each $x_i(t)-x_{\textup r,t}(t)$,  we introduce the following coordinate transformations for the reference system \cref{eq:ref-system} and the nominal system \cref{eq:nominal-cl-system} for each $i\in\Zn$:
\begin{equation}\label{eq:coordinate-trans}
\left \{
    \begin{aligned}
        \check x_\rt & = \Tau_x^ix_\rt,\quad \check x_{\nt} = \Tau_x^i \xn, \\
         \check A_m^i &= \Tau_x^iA_m (\Tau_x^i)^{-1},\\
         \check B^i  & = \Tau_x^iB, \quad \check B^i _v = \Tau_x^iB_v,   
    \end{aligned}
    \right.
\end{equation}
where $\Tau_x^i\!>\!0$ is a diagonal matrix that satisfies 
\begin{align}
    \Tau_x^i[i]&=1, \ 0<\Tau_x^i[k]\leq 1, \ \forall k\neq i, \label{eq:Tx-i-cts}
\end{align}
with $\Tau_x^i[k]$ denoting the $k$th diagonal element. Under the transformation \cref{eq:coordinate-trans}, the reference system \cref{eq:ref-system} is converted to 
\begin{equation}\label{eq:ref-system-transformed}
\hspace{-2mm}
\left\{
\begin{aligned}
  \dot {\check x}_\rt(t) & \!=\! \check A_m^i\check x_\rt(t)\! +\!  \check  B_v^iv(t) \!+\!  \check B^i (u_{\rt}(t) \!+\! \!\check f(t, \check x_\rt(t))),  \\ 
    {u_{\textup{r}}}(s) & \! =\! -\mcC(s)\laplace{\check f(t,\check x_\textup{r}(t))},  \ \check x(0) \!=\! \Tau_x^ix_0,
\end{aligned}\right. 
\end{equation}
where 
\begin{equation}\label{eq:f-checkf-relation}
    \check f(t,\check x_\rt(t)) =  f(t, x_\rt(t))) =  f(t, (\Tau_x^i)^{-1}\check x_\rt(t))).
\end{equation}Given a set $\mcZ$, define 
\begin{equation}\label{eq:check-Z-defn}
    \check \mcZ\trieq \{\check z\in \mbR^n: (\Tau_x^i)^{-1}\check z \in \mcZ\}.
\end{equation}
Similar to \cref{eq:Hxm-Hxv-Gxm-defn}, for the transformed reference system \cref{eq:ref-system-transformed}, we have
\begin{subequations}\label{eq:Hxm-Hxv-Gxm-check-defn}
\begin{align}
  \mcH_{\check xm}^i(s) &\trieq (sI_n \!-\! \check A_m^i)^{-1}\! \check B^i  = \Tau_x^i \mcH_{xm}(s), \label{eq:Hxm-check-defn}  \\
     \mcH_{\check xv}^i(s) & \trieq (sI_n \!-\! \check A_m^i)^{-1} \! \check B^i _v = \Tau_x^i \mcH_{xv}(s), \label{eq:Hxv-check-defn} \\
  \mcG_{\check xm}^i(s) & \trieq \mcH_{\check xm}^i(s)(I_m-\mcC(s)) =  \Tau_x^i \mcG_{xm}(s), \label{eq:Gxm-check-defn}
\end{align}
\end{subequations}
where $\mcH_{xm},~\mcH_{xv},~\mcG_{xm}$ are defined in \cref{eq:Hxm-Hxv-Gxm-defn}. By applying the transformation \cref{eq:coordinate-trans} to the nominal system \cref{eq:nominal-cl-system}, we obtain 
\begin{equation}\label{eq:nominal-cl-system-transformed}
\hspace{-2mm}
\left\{
\begin{aligned}
  \dot {\check x}_{\textup n}(t) & = \check A_m^i\xcheckn(t) + \check B^i _vv(t),\ \xcheckn(0) \!=\! \Tau_x^ix_0,  \\ 
  \yn(t) & = \check C \xcheckn(t).
\end{aligned}\right. 
\end{equation}
Letting $\xcheckin(t)$ be the state of the system 
    $\dot{\check x}_\textup{in}(t) = \check A_m^i \xcheckin(t)$ with $\xcheckin(0) = \check x_\nt(0) = \Tau_x^i x_0$, 
we have $\xcheckin(s) \trieq (sI_n-\check A_m^i)^{-1} \xcheckin(0) = \Tau_x^i(sI_n- A_m)^{-1} x_0$. Defining
\begin{equation}\label{eq:check-rhoin-defn}
   \check \rho_\textup{in}^i \trieq \lonenorm{s\Tau_x^i(sI_n- A_m)^{-1}}\max_{x_0\in \mcX_0}\infnorm{x_0},
\end{equation}
and further considering \cref{lem:L1-Linf-relation}, we have 
    $\linfnorm{\xcheckin}\leq \check \rho_\textup{in}^i$.
Similar to \cref{eq:l1-stability-condition}, for the transformed reference system \cref{eq:ref-system-transformed}, consider the following condition:
\begin{align}
\lonenorm{\mcG_{\check xm}^i(s)}b_{\check f,\check \mcX_r}  < \check \rho_r^i - \lonenorm{\mcH_{\check x v}^i(s)}\linfnorm{v} - \check \rho_\textup{in}^i ,\label{eq:l1-stability-condition-transformed} %
\end{align}
where 
$\mcX_r$ is defined in \cref{eq:Xr-Xa-defn} and $\check \mcX_r$ is defined according to \cref{eq:check-Z-defn} and $\check \rho_r^i$ is a positive constant to be determined. 
Then we have the following result. 

\begin{lemma}\label{lem:refine_bnd_xi_w_Txi}
Consider the reference system \eqref{eq:ref-system} 
subject to \cref{assump:lipschitz-bnd-fi}, the nominal system \cref{eq:nominal-cl-system}, the transformed reference system \cref{eq:ref-system-transformed} and transformed nominal system \cref{eq:nominal-cl-system-transformed} obtained by applying \cref{eq:coordinate-trans} with any $\Tau_x^i$ satisfying \cref{eq:Tx-i-cts}. 
Suppose that \cref{eq:l1-stability-condition} holds with some constants $\rho_r$ and $\linfnorm{v}$.
Then, 
there exists an constant  $\check \rho_r^i\leq \rho_r$ such that \cref{eq:l1-stability-condition-transformed}
holds with the same $\linfnorm{v}$. Furthermore,  
\begin{align}
\abs{x_{\rt,i}(t)} & \leq  \check \rho_r^i, \  \forall t\geq0, \label{eq:xr-i-bnd-from-trans} \\
    \abs{x_\textup{r,i}(t)-x_{\nt,i}(t)} & \leq \lonenorm{\mcG_{ \check xm}(s)}b_{f,\mcX_r},\ \forall t\geq 0, \label{eq:xri-xni-bnd-from-trans}
\end{align}
where we re-define \begin{equation}\label{eq:Xr-defn}
   \mcX_r\trieq \left\{z\in\mbR^n: \abs{z_i}\leq \check \rho_r^i, i\in\Zn\right\}. 
\end{equation}
\end{lemma}
\begin{proof}
For any $\Tau_x^i$ satisfying \cref{eq:Tx-i-cts} with an arbitrary $i\in\Zn$, we have $\infnorm{\Tau_x^i}= 1$. Therefore, under the transformation \cref{eq:coordinate-trans}, considering \cref{eq:Hxm-Hxv-Gxm-check-defn,eq:check-rhoin-defn} and \cref{lem:L1-Linf-relation}, we have 
\begin{subequations}\label{eq:Hxm-Hxv-Gxm-check-original-relation}
\begin{align}
      \hspace{-2mm}   \linfnorm{\mcH_{\check xm}^i(s)} & \! \leq \! \infnorm{\Tau_x^i} \linfnorm{\mcH_{xm}(s)} \!=\! \linfnorm{\mcH_{xm}(s)}\!,  \label{eq:Hxm-check-original-relation} \\
   \hspace{-2mm}   \linfnorm{\mcH_{\check xv}^i(s)} & \! \leq \! \infnorm{\Tau_x^i}  \linfnorm{\mcH_{xv}(s)} \!=\!  \linfnorm{\mcH_{xv}(s)}\!, \label{eq:Hxv-check-original-relation}\\
\hspace{-2mm}  \linfnorm{\mcG_{\check xm}^i(s)} & \! \leq \!  \infnorm{\Tau_x^i}  \linfnorm{\mcG_{xm}(s)} \!=\!  \linfnorm{\mcG_{xm}(s)}\!, \label{eq:Gxm-check-original-relation} \\
\hspace{-2mm}   \check \rho_\textup{in}^i & \! \leq \! \infnorm{\Tau_x^i}  \rhoin \!=\! \rhoin. \label{eq:rhoin-check-original-relation}
\end{align}
\end{subequations}
It follows from \cref{lem:ref-xr-ur-bnd} that $\xr(t)\in \mcX_r$ for any $t\geq0$, which, together with \cref{eq:coordinate-trans}, implies $\check x_\rt(t)\in \check \mcX_r$ for any $t\geq0$, where $\check \mcX_r$ is defined via \cref{eq:check-Z-defn}. 
Considering \cref{eq:f-checkf-relation,eq:check-Z-defn}, for any compact set $\mcX_r$, we have
\begin{equation}\label{eq:b-checkf-f-equal}
    b_{\check f, \check \mcX_r} = b_{f, \mcX_r}.
\end{equation}
Now suppose that constants $\rho_r$ and $\linfnorm{v}$ satisfy \cref{eq:l1-stability-condition}.  
Then, due to \cref{eq:b-checkf-f-equal,eq:Hxm-Hxv-Gxm-check-original-relation}, with $\check \rho_r^i = \rho_r$ and the same $\linfnorm{v}$, \cref{eq:l1-stability-condition-transformed} is satisfied. 

Additionally, if \cref{eq:l1-stability-condition-transformed} holds, by applying \cref{lem:ref-xr-ur-bnd} to the transformed reference system \cref{eq:ref-system-transformed}, we obtain that $\linfnorm{\check x_\rt}\leq \check \rho_r^i$, implying that $\abs{\check x_{\rt,i}(t)}\leq \check \rho_r^i$ for any $t\geq 0$. Since $\check x_{\rt,i}(t) =x_{\rt,i}(t) $ due to the constraint \cref{eq:Tx-i-cts} on $\Tau_x^i$, we have \cref{eq:xr-i-bnd-from-trans}.
Equation \cref{eq:xr-i-bnd-from-trans} is equivalent to $\xr(t)\in\mcX_r$ for any $t\geq 0$, with the re-definition of $\mcX_r$ in \cref{eq:Xr-defn}. 
Following the proof of \cref{lem:ref-id-bnd}, one can obtain $   \left\| {\check x_\textup{r} - \check x_\nt} \right\|_{{\mathcal L}{_\infty }}  \leq  \lonenorm{\mcG_{ \check xm}} b_{\check f, \check \mcX_r}= \lonenorm{\mcG_{ \check xm}(s)}b_{f,\check \mcX_r}$, where  the equality is due to \cref{eq:b-checkf-f-equal}. Further considering  $\check x_{\rt,i}(t) =x_{\rt,i}(t) $ and $\check x_{\nt,i}(t) = x_{\nt,i}(t)$  due to the constraint \cref{eq:Tx-i-cts} on $\Tau_x^i$, we have \cref{eq:xri-xni-bnd-from-trans}.
\end{proof}
\begin{remark}
\cref{lem:ref-xr-ur-bnd} and \cref{lem:ref-id-bnd} imply  $\abs{x_{\rt,i}(t)}\leq \rho_r$ and $\abs{x_{\rt,i}(t)-x_{\nt, i}(t)}\leq \lonenorm{\mcG_{xm}}b_{f,\Omega(\rho_r)}$, respectively, for all $i\in\Zn$ and $t\geq 0$. \cref{lem:refine_bnd_xi_w_Txi} indicates that by applying the coordinate transformation \cref{eq:coordinate-trans} and leveraging the condition \cref{eq:l1-stability-condition-transformed} for the transformed system \cref{eq:ref-system-transformed}, one can obtain a tighter bound on $x_{\rt,i}(t)$ than $\rho_r$ and a tighter bound on $\abs{x_{\rt,i}(t)-x_{\nt, i}(t)}$ than $\lonenorm{\mcG_{xm}(s)}b_{f,\Omega(\rho_r)}$. 
\end{remark}
\noindent \textbf{Derive Separate Bounds for Adaptive Inputs}:
From  \cref{eq:ref-system-uar} and the structure with $\mcC(s)$  \cref{eq:filter-defn}, we can obtain  \begin{equation}
    u_{\rt,j}(s)=-\mcC_j(s) \laplace{f_j(t,\xr(t))}, \quad \forall j\in\Zm.
\end{equation}
Therefore, given a set $\mcX_r$ such that $\xr(t)\!\in\! \mcX_r$ for any $t\!\geq\! 0$, from \cref{assump:lipschitz-bnd-fi,lem:L1-Linf-relation} we get 
\begin{equation}\label{eq:uar-i-bnd}
 \abs{u_{\rt,j}(t)} \leq \lonenorm{\mcC_j(s)} b_{f_j, \mcX_r}, \quad \forall t\geq 0, \ \forall j\in\Zm.
\end{equation}

With the preceding preparations, we are ready to {derive an individual bound for $x_i(t)-x_{\nt,i}(t)$ ($i\in\Zn$) and $u_j(t)-u_{\nt,j}(t)$} ($j\in\Zm$), as stated in the following theorem.
\begin{theorem}\label{them:xi-uai-bnd}
Consider the uncertain system \eqref{eq:dynamics-uncertain} subject to \cref{assump:lipschitz-bnd-fi}, the nominal system \eqref{eq:nominal-cl-system}, and the \loneAC~defined via \cref{eq:state-predictor,eq:adaptive_law,eq:l1-control-law} subject to the conditions \cref{eq:l1-stability-condition-Lf,eq:l1-stability-condition} with constants $\rho_r$ and $\gamma_1>0$ and the sample time constraint \eqref{eq:T-constraint}. Suppose that for each $i\in \Zn$, 
\cref{eq:l1-stability-condition-transformed} holds with a constant $\check \rho_r^i$ for the transformed reference system \cref{eq:ref-system-transformed} obtained by applying \cref{eq:coordinate-trans}. Then, we have 
\begin{subequations}\label{eq:xui-xuni-bnd-from-trans-w-tilX-tilU-defn}
\begin{align}
\hspace{-2mm}     {x(t)-x_{\nt}(t)} \!\in\! \tilde \mcX\!\trieq\! \left\{z\!\in\!\mbR^n\!: \!\abs{z_i}\!\leq\! \tilde \rho^i, \  i\in \Zn \right\}\!,\  & \forall t\!\geq\! 0, \label{eq:xi-xni-bnd-from-trans-w-tilX-defn}\\
 \hspace{-2mm} {u_{\at}(t)} \!\in\! \mcU_\at  \!\trieq\! \left\{z\!\in\!\mbR^m\!: \!\abs{z_j}\!\leq\! \rho_{u_\at}^j, \ \! j\!\in\! \Zm \right \}\!,\ &\forall t\!\geq\! 0, \label{eq:ua-i-bnd-w-Ua-defn} \\ 
 \hspace{-2mm}        u(t) - u_{\nt}(t)\!\in\!  \tilde\mcU   \!\trieq\! \left \{z \!\in\! \mbR^m\!: \!\abs{z_j}\!\leq\!  \tilde\rho_u^j,
\  j\!\in\! \Zn \right \}\!,  \ & \forall t\!\geq\!0, \label{eq:uj-unj-bnd-w-tilU-defn} 
\end{align}
\end{subequations}where 
\begin{subequations}\label{eq:rhoi-tilrhoi-y-defn}
\begin{align}
 \hspace{-2mm}      \rho^i &  \!\trieq\! \check \rho_r^i\!+\!\gamma_1, \ \tilde \rho^i\!\trieq\! \lonenorm{\mcG_{ \check xm}^i(s)}\!\!b_{f,\mcX_r}\!+\!\gamma_1,\label{eq:rhoi-tilrhoi-defn}
\\
\hspace{-2mm}\rho_{u_\at}^j & \!\trieq\! \lonenorm{\mcC_j(s)} b_{f_j, \mcX_r}\! +\! \gamma_2,\   \tilde\rho_u^j \!\trieq \!  \rho_{u_\at}^j \!+\!
\sum_{i=1}^n\abs{K_x[j,i]}\tilde \rho^i,
\label{eq:tilrho-u-j-defn}
\end{align}
\end{subequations}
with $\mcX_r$ defined in \cref{eq:Xr-defn}, and $C[j,i]$ denoting the $(j,i)$ element of $C$.

\end{theorem}
\begin{proof}
For each $i\in \Zn$, \cref{lem:refine_bnd_xi_w_Txi} implies 
$\abs{x_{\rt,i}(t)} \leq \check \rho_r^i$ and $\abs{x_\textup{r,i}(t)-x_{\nt,i}(t)} \leq \lonenorm{\mcG_{ \check xm}^i(s)}b_{f,\mcX_r}$ for all  $t\geq0$.  On the other hand, \cref{them:x-xref-bnd} indicates that $\abs{x_{\rt,i}(t) - x_{i}(t) } \leq \gamma_1$ for any $t\geq 0$ and any $i\in \Zn$. Therefore, \cref{eq:xi-xni-bnd-from-trans-w-tilX-defn} is true. 
On the other hand, \cref{them:x-xref-bnd} indicates that $\abs{u_{\rt,j}(t) - u_{\at,j}(t) } \leq \gamma_2$ for any $t\geq 0$ and any $j\in\Zm$, which, together with \cref{eq:uar-i-bnd}, leads to \cref{eq:ua-i-bnd-w-Ua-defn}. Finally, \cref{eq:uj-unj-bnd-w-tilU-defn} follows from \cref{eq:u-un-defn,eq:ua-i-bnd-w-Ua-defn,eq:xi-xni-bnd-from-trans-w-tilX-defn}. The proof is complete. 
\end{proof}
\begin{remark}\label{rem:x-xid-discussion-Cs-T-refined}
\cref{them:xi-uai-bnd} provides a way to derive an individual bound on $x_i(t)$, and $x_i(t)-x_{\nt,i}(t)$ for each $i\in \Zn$ and on $u_j (t) - u_{\nt,j}(t)$ for each $j\in\Zm$ via coordinate transformations. 
Additionally, similar to the arguments in \cref{rem:x-xid-discussion-Cs-T}, by decreasing $T$ and increasing the bandwidth of the filter $\mcC(s)$, one can make $\tilde \rho^i$ ($i\in \Zn$) arbitrarily small, i.e., making {\it the states of the adaptive system arbitrarily close to those of the nominal system}, and make the bounds on $u_{\at,j}(t)$ and $u_j(t)-u_{\nt,j}(t)$ arbitrarily close to the bound on the true uncertainty $f_j(t,x)$ for $x\in\mcX_a$, for each $j\in\Zm$.
\end{remark}
According to \cref{them:x-xid-bnd,them:xi-uai-bnd}, the procedure for designing an \loneAC~with separate bounds on states and adaptive inputs can be summarized in \cref{alg:l1ac-design-w-sep-bnds}. 
\begin{algorithm}[h]
 \caption{Designing an \loneAC~with separate bounds}\label{alg:l1ac-design-w-sep-bnds}
\begin{algorithmic}[1]
\Require{uncertain system \cref{eq:dynamics-uncertain} subject to \cref{assump:lipschitz-bnd-fi}, initial parameters $A_e$, $\mcC(s)$ and $T$ to define an \loneAC,  $\gamma_1$, $\mcX_0$, $\linfnorm{v}$, tol}
\Procedure{DecideFilterUncertBnd}{$\mcC(s)$,$\gamma_1$,$\mcX_0$,$\linfnorm{v}$}\label{procedure:l1ac-design}
\While{condition \cref{eq:l1-stability-condition} or \cref{eq:l1-stability-condition-Lf} does not hold} \label{line:l1-stability-conditions-in-l1ac}
\State Increase the bandwidth of $\mcC(s)$ \Comment{See \cref{rem:Gxm_linfnorm-stability-filter}.} 
\EndWhile \Comment{$\rho_r$, $\mcX_r \!=\! \Omega(\rho_r)$ and $b_{f,\mcX_r}$ will be computed.}
\EndProcedure
\State Set $b_{f,\mcX_r}^{old} = b_{f,\mcX_r}$ \label{line:b_fXr-old}
\Procedure{DeriveSepStateBnds}{$b_{f,\mcX_r}$,$\gamma_1$,$\mcC(s)$,$\mcX_0$,$\linfnorm{v}$}\label{procedure:sep-state-bnds}
\For{$i=1,\dots,n$}
\State Select $\Tau_x^i$ satisfying\!~\cref{eq:Tx-i-cts} and apply the transformation\!~\cref{eq:coordinate-trans}\label{line:Tx_i-selection}
\State Evaluate \cref{eq:Hxm-Hxv-Gxm-check-defn} and 
compute $\check \rho_\textup{in}^i$ according to \cref{eq:check-rhoin-defn}
\State Compute $\check \rho_{r}^i$ that satisfies \cref{eq:l1-stability-condition-transformed} \Comment{Such a $\check \rho_r^i\leq \rho_r$ is guaranteed to exist from \cref{lem:refine_bnd_xi_w_Txi}}
\State Set $\rho^i =  \check \rho_r^i+\gamma_1$, $\tilde \rho^i = \lonenorm{\mcG^i_{ \check xm}(s)}b_{f,\mcX_r}+\gamma_1$
\EndFor
\EndProcedure
\State Set $\mcX_r\! = \!\left\{\!z\!\in\!\mbR^n\!:\! \abs{z_i}\!\leq\! \check\rho_r^i\right\}$ and update $b_{f,\mcX_r}$\label{line:Xr-update-in-l1ac} 

\If{$b_{f,\mcX_r}^{old} - b_{f,\mcX_r}>\textup{tol}$}
\State Set $b_{f,\mcX_r}^{old} = b_{f,\mcX_r}$ and go to step~\ref{procedure:sep-state-bnds}
\EndIf
\State Set $\rho = \max_{i\in\Zn} \rho^i$ and compute $\mcX_a$ via \cref{eq:Xr-Xa-defn}\label{line:Xa-update-in-l1ac} 
\Procedure{DeriveSepInputBnds}{$\mcX_r,\{\tilde \rho^i\}_{i\in\Zn},\gamma_2,\mcC(s)$}\label{procedure:sep-input-bnds}
\For{$j=1,\dots,m$}
\State Compute $\rho_{u_\at}^j$ and $\tilde \rho_u^j$ according to \cref{eq:tilrho-u-j-defn}
 \EndFor 
 \EndProcedure
 \Procedure{DecideSampleTime}{$A_e$,$\mcC(s)$,$T$,$\mcX_a$} \label{line:decide-sample-time}
\While{constraint \cref{eq:T-constraint} does not hold}
\State Decrease $T$ \Comment{Small $T$ can enforce \cref{eq:T-constraint} due to \cref{eq:gammaT-to-0}.}
\EndWhile
\EndProcedure
\Ensure{An \loneAC~defined by \cref{eq:l1-control-law,eq:adaptive_law,eq:state-predictor}~with parameters $A_e$ and $\mcC(s)$ and $T$, $\rho^i$ and $\tilde \rho^i$ for $i
\in \Zn$, $\rho_{u_\at}^j$ and $\tilde \rho_{u}^j$ for $j \in \Zm$ } 
\end{algorithmic}
\end{algorithm}
\begin{remark}
One can try different $\Tau_x^i$ in step~\ref{line:Tx_i-selection} of \cref{alg:l1ac-design-w-sep-bnds} and select the one that yields the tightest bound for the $i$th state.
\end{remark}

\begin{remark}\label{rem:conservatism-in-Cs-T}
The conditions \cref{eq:T-constraint,eq:l1-stability-condition-w-extra-Lf} can be quite conservative for some systems, due to the frequent use of inequalities related to the \lonew norm (stated in \cref{lem:L1-Linf-relation}), Lipschitz continuity and matrix/vector norms. As a result, the bandwidth of the filter $\mcC(s)$ computed via \cref{eq:l1-stability-condition-w-extra-Lf} could be unnecessarily high, while the sample time $T$ computed via \cref{eq:T-constraint} under a given $\gamma_1$ could be unnecessarily small. Based on our experience, assuming that some bounds $\tilde \rho^i$ ($i\in\Zn$) and $\tilde \rho_u^j$ ($j\in\Zm$) satisfying \cref{eq:xui-xuni-bnd-from-trans-w-tilX-tilU-defn} are derived under a specific filter $\mcC^\star(s)$ and $T^\star$ that satisfy \cref{eq:T-constraint,eq:l1-stability-condition-w-extra-Lf}, those bounds will most likely be respected {\it in simulations} even if we decrease the bandwidth of $\mcC^\star(s)$ by $1\sim3$ times and/or increase $T^\star$ by $1\sim 10$ times.
\end{remark}

 \section{\loneRG: Adaptive Reference Governor for Constrained Control Under Uncertainties}\label{sec:l1rg}
 Leveraging the uniform bounds on state and input errors guaranteed by the \loneAC, we now integrate the \loneAC~and the RG introduced in \cref{sec:sub-rg-nominal} to synthesize the \loneRG~framework for simultaneously enforcing the constraints \cref{eq:constraints} and improving the tracking performance.  
\subsection{\loneRG~Design}\label{sec:sub-l1rg}
We first make the following  assumption.
\begin{assumption}\label{assump:X0-Xn-xn}
$\hat \mcX_\nt$ and $\hat \mcU_\nt$ defined by \cref{eq:Xn-Un-defn}, \cref{eq:Xn-hat-Un-hat-defn} and \cref{eq:xui-xuni-bnd-from-trans-w-tilX-tilU-defn}  are nonempty. Furthermore, there exists a known command $v(0)$ such that 
\begin{equation}\label{eq:x0-0-in-tilOinf}
    (v(0),x_0)\in \tilOinf,
\end{equation}
where $\tilOinf$ is defined in \cref{eq:tilOinf-defn}. \end{assumption}
\begin{remark}
Considering \cref{eq:tilOinf-defn},  \cref{eq:x0-0-in-tilOinf} implies $x_0\in \hat \mcX_\nt$ and  $u_\nt(0) =K_xx_0+K_vv(0)\in \hat \mcU_\nt$ (since $\xn(0)=x_0$) where $\hat \mcX_\nt$ and $\hat \mcU_\nt$, according to \cref{eq:Xn-hat-Un-hat-defn}, are tightened versions of $\mcX_\nt$ and $\mcU_\nt $ that are defined in \cref{eq:Xn-Un-defn}. 
From \cref{rem:x-xid-discussion-Cs-T-refined}, with a sufficiently high bandwidth for $\mcC(s)$ and sufficiently small $T$, one can make $\mcX_\nt$ arbitrarily close to $\mcX$,
 and make $\tilde \mcU$ arbitrarily close to the bound set for the true uncertainty in $\mcX$. Additionally, as mentioned in 
 \cref{rem:Xn-hat-Xn-relation-w-Td}, $\hat \mcX_\nt$ and $\hat \mcU_\nt$ are close to $\mcX_\nt$ and $ \mcU_\nt$, respectively, when $T_d$ is small. As a result, with a sufficiently high bandwidth for $\mcC(s)$, and sufficiently small $T$ and $T_d$, \cref{assump:X0-Xn-xn} roughly states that the initial state stays in $\mcX$, and the constraint set $\mcU$ is sufficiently large to ensure enough control authority for tracking an initial reference command $v(0)$ and additionally for compensating the uncertainty in $ \mcX$.
\end{remark}
Under the preceding assumption, the design procedure for \loneRG~is summarized in \cref{alg:l1rg}.  
Compared to step~\ref{line:l1-stability-conditions-in-l1ac} of \cref{alg:l1ac-design-w-sep-bnds},
we additionally constrain $\xr(t)$ and $x(t)$  to stay in $\mcX$ for all $t\geq0$
in step~\ref{line:l1-stability-conditions-in-l1rg} of \cref{alg:l1rg}. 
Such constraints 
can potentially limit the size of uncertainties that need to be compensated and significantly reduce the conservatism of the proposed solution. 
\begin{algorithm}
\caption{\loneRG~Design}\label{alg:l1rg}
\begin{algorithmic}[1]
\Require{An continuous-time uncertain system \cref{eq:dynamics-uncertain} subject to \cref{assump:lipschitz-bnd-fi}, constraint sets $\mcX$ and $\mcU$ as in \cref{eq:constraints}, $\mcX_0$, $\mcV$ (admissible set for $v(t)$), baseline control law in \cref{eq:baseline-control-law}, initial parameters $A_e$, $\mcC(s)$ and $T$ to define an \loneAC, $\gamma_1$, $T_d$ and $\epsilon$ for RG design, tol}
\Procedure{\loneAC-DesignUnderConstraints}{}
\State Compute $\linfnorm{v}$ given $\mcV$
\While{\cref{eq:l1-stability-condition} with $\mcX_r = \Omega(\rho_r) \cap \mcX$ or \cref{eq:l1-stability-condition-Lf} with $\mcX_a = \Omega(\rho_r+\gamma_1) \cap \mcX$
does not hold with any $\rho_r$}\label{line:l1-stability-conditions-in-l1rg}
\State Increase the bandwidth of $\mcC(s)$ \Comment{See \cref{rem:Gxm_linfnorm-stability-filter}.} 
\EndWhile \Comment{$\mcX_r$ and $b_{f,\mcX_r}$  will be computed.} \label{line:Xr-b_f-under-csts}
\State Set $b_{f,\mcX_r}^{old} = b_{f,\mcX_r}$
\State Run \textbf{\sc DeriveSepStateBnds}  of \cref{alg:l1ac-design-w-sep-bnds} with  $b_{f,\mcX_r}$, 
and obtain $\check \rho_r^i$
$\rho^i$ and $\tilde \rho^i$ for $i\in \Zn$\label{line:run-sep-state-bnds-in-l1rg}\label{line:derive-sep-state-bnds-in-l1rg}
\State Set $\mcX_r\! = \!\left\{\!z\!\in\!\mbR^n\!:\! \abs{z_i}\!\leq\! \check\rho_r^i\right\}\cap \mcX$ and update $b_{f,\mcX_r}$ \label{line:Xr-update-in-l1rg}
\If{$b_{f,\mcX_r}^{old} - b_{f,\mcX_r}>\textup{tol}$}
\State Set $b_{f,\mcX_r}^{old} = b_{f,\mcX_r}$ and go to step~\ref{procedure:sep-state-bnds}
\EndIf
\State Set $\mcX_a = \{ z\in\mbR^n: \abs{z_i}\leq \rho^i,\ i\in\Zn\}\cap \mcX$ \label{line:Xa-update-in-l1rg}
\State Run \textbf{\sc DeriveSepInputBnds}  of \cref{alg:l1ac-design-w-sep-bnds} with $\mcC(s)$ from step~\ref{line:Xr-b_f-under-csts}  and $\mcX_r$ from step~\ref{line:Xr-update-in-l1rg}, and obtain $ \rho_{u_\at}^j$ and $\tilde \rho_u^j$ for $j\in \Zm$ \label{line:derive-sep-input-bnds-in-l1rg}
\State Run \textbf{\sc DecideSampleTime}  of \cref{alg:l1ac-design-w-sep-bnds} with $\mcC(s)$ from step~\ref{line:Xr-b_f-under-csts}  and $\mcX_a$ from step~\ref{line:Xa-update-in-l1rg}, and obtain $T$
\State Compute $\tilde \mcX$ and $\tilde \mcU$ with $\{\tilde \rho^i\}_{i\in\Zn}$ and $\{\tilde \rho_u^j\}_{j\in\Zm}$ 
via \cref{eq:xui-xuni-bnd-from-trans-w-tilX-tilU-defn}
\EndProcedure
\Procedure{RG-Design}{}\label{line:rg-design}
\State Compute $\mcX_\nt$ and $\mcU_\nt$ with $\mcX$,~$\mcU$,~$\tilde\mcX$ and $\tilde\mcU$ via \cref{eq:Xn-Un-defn} 
\State Formulate the nominal discrete-time model \cref{eq:nominal-system-discrete-w-un} with the sample time $T_d$ 
\State Compute $\hat \mcX_\nt$ and $\hat \mcU_\nt$ via \cref{eq:Xn-hat-Un-hat-defn} \hspace{-4mm}  \Comment{With a small $T_d$, one may set $\hat \mcX_\nt \!= \!\mcX_\nt$, $\hat \mcU_\nt \!= \! \mcU_\nt $ for practical implementation.}
\State Compute the set $\tilOinf$ dependent on $\epsilon$ according to \cref{eq:tilOinf-defn}
\EndProcedure
\Ensure{An \loneRG~consisting of a RG designed for the nominal system \cref{eq:nominal-cl-system} and an \loneAC~to compensate for uncertainties}
\end{algorithmic}
\end{algorithm}

We are ready to state the guarantees regarding tracking performance and constraint enforcement provided by \loneRG.
\begin{theorem}\label{them:xi-uai-bnd-l1rg}
Consider an uncertain system \eqref{eq:dynamics-uncertain} subject to \cref{assump:lipschitz-bnd-fi} and the state and control constraints in \cref{eq:constraints}.  
Suppose that an \loneAC~(defined by \cref{eq:l1-control-law,eq:adaptive_law,eq:state-predictor}) and a RG 
are designed by following \cref{alg:l1rg}.  
If \cref{assump:X0-Xn-xn} hold,
then, under the baseline control law \cref{eq:baseline-control-law} and the \loneRG~consisting of the compositional control law \cref{eq:total-control-law}, the \loneAC~and the RG for computing the reference command $v(t)$ according to \cref{eq:v-t-kTd-relation,eq:opt-prob-rg}, we have 
\begin{align}
    x(t)\in \tint(\mcX), \  u(t) \in \tint(\mcU),\quad &\forall t\geq 0,  \label{eq:constraints-xu} \\
    x(t)-x_{\nt}(t) \in  \tint(\tilde \mcX),\quad & \forall t\geq 0,\label{eq:xi-xni-bnd-under-l1rg}
    \\
    y(t)-y_{\nt}(t)\in \{z\in\mbR^m: \abs{z_i}\leq \tilde \rho_y^j\},\quad  &\forall t\geq 0,  
   \label{eq:yi-yni-bnd-under-l1rg}
\end{align}
where $\xn(t)$ and $\yn(t)$ are the states and outputs of the nominal system \cref{eq:nominal-cl-system-w-un} under the reference command input $v(t)$, and 
\begin{equation}\label{eq:tilrhoi-y-defn}
    \tilde \rho_y^j  \!\trieq\! \sum_{i=1}^n \abs{C[j,i]}\tilde\rho^i,\quad \forall j\in \Zm. 
\end{equation}

\begin{proof}
Equation \cref{eq:x0-0-in-tilOinf} in \cref{assump:X0-Xn-xn} implies $(\mathrm v(0),\mathrm  x_\nt(0))\in \tilOinf$ (due to $\mathrm \xn(0)=x_0$),
and $\mathrm u_\nt(0) =K_x\mathrm \xn(0)+K_v\mathrm  v(0)\in \hat \mcU_\nt$. Thus, 
the reference command $\mathrm v(k)$ produced by  \cref{eq:opt-prob-rg} ensures  $\mathrm x_\nt(k)\in\hat \mcX_\nt$  and $\mathrm u_\nt(k)\in\hat\mcU_\nt$ for all $k\in \mbZ_+$, which, due to \cref{lem:nominal-system-inter-sample-behavior}, implies
\begin{equation}\label{eq:constraints-nom-in-l1rg-sec}
    \xn(t) \in \mcX_\nt, \ u_\nt(t)\in \mcU _\nt,\quad \forall t\geq 0.
\end{equation}
Compared to \cref{line:l1-stability-conditions-in-l1ac,line:Xr-update-in-l1ac,line:Xa-update-in-l1ac} of \cref{alg:l1ac-design-w-sep-bnds}, we restrain $\mcX_r$ and $\mcX_a$ to be subsets of $\mcX$ in \cref{line:l1-stability-conditions-in-l1rg,line:Xr-update-in-l1rg,line:Xa-update-in-l1rg} of  \cref{alg:l1rg}. As a result,  {\it if}  \cref{eq:constraints-xu} and 
\begin{equation}\label{eq:xr-in-X}
\xr(t)\in\mcX,\quad \forall t\geq 0,
\end{equation}
jointly hold, condition \cref{eq:xi-xni-bnd-under-l1rg} holds according to \cref{them:xi-uai-bnd}, while \cref{eq:xri-xni-bnd-from-trans} holds according to 
\cref{lem:refine_bnd_xi_w_Txi}.


We next prove \cref{eq:constraints-xu,eq:xr-in-X} by contradiction. Assume \cref{eq:constraints-xu} or \cref{eq:xr-in-X}  do not hold. The initial condition \cref{eq:x0-0-in-tilOinf} implies that 
$x_0\in \mcX_\nt\subset \tint(\mcX)$ and $u_\nt(0) \in \mcU_\nt\subset \tint(\mcU)$. 
As a result, we have $x(0)\in\mcX$, $\xr(0)\in \mcX$ and $u(0)\in \mcU$. Since $x(t)$, $\xr(t)$ and $u(t)$ are continuous, there must exist a time instant $\tau$, such that \begin{subequations}\label{eq:x-0-tau-w-tau}
\begin{align} 
\hspace{-3mm}x(t) &\!\in\! \textup{int}(\mcX),\ \xr(t)\!\in\! \textup{int}(\mcX)  \textup{ and } u(t)\!\in\! \textup{int}(\mcU), \ \forall t \!\in\! [0,\tau)  \label{eq:x-0-tau} \\
 \hspace{-3mm}   x(\tau) & \!\in\! \textup{bnd}(\mcX) \textup{ or }  \xr(\tau) \!\in\! \textup{bnd}(\mcX)  \textup{ or }  u(\tau) \!\in\! \textup{bnd}(\mcU). \label{eq:x-tau]}
\end{align}    
\end{subequations}
Now consider the interval $[0,\tau]$. According to \cref{lem:refine_bnd_xi_w_Txi}, due to \cref{eq:x-0-tau-w-tau} and the definitions in \cref{eq:xi-xni-bnd-from-trans-w-tilX-defn,eq:rhoi-tilrhoi-defn}, we have $\xr(t)-\xn(t) \in \{z\in\mbR^n: \abs{z_i}\leq \lonenorm{\mcG_{ \check xm}(s)}b_{f,\mcX_r}\}\subset \tint(\tilde \mcX)$, which, together with \cref{eq:Xn-Un-defn,eq:constraints-nom-in-l1rg-sec}, implies 
\begin{equation}\label{eq:xr-in-Xinterior-under-l1rg}
    \xr(t)\in\tint(\mcX),\quad \forall t\in [0,\tau].
\end{equation}
Similarly, according to \cref{them:xi-uai-bnd}, due to \cref{eq:x-0-tau-w-tau} and the definition in \cref{eq:xui-xuni-bnd-from-trans-w-tilX-tilU-defn}, we have $
    x(t)-\xn(t)  \in \interior{\tilde \mcX}$, and $u(t)- \un(t) \in \interior{\tilde \mcU}$, for any $ t \in [0,\tau]$, 
which, together with \cref{eq:Xn-Un-defn,eq:constraints-nom-in-l1rg-sec}, implies
\begin{equation}\label{eq:xu-in-XUinterior-under-l1rg}
   x(t)\in \interior{\mcX},\ u(t) \in \interior{\mcU}, \quad \forall t\in [0,\tau].
\end{equation} Both \cref{eq:xr-in-Xinterior-under-l1rg} and \cref{eq:xu-in-XUinterior-under-l1rg} contradict  \cref{eq:x-tau]}, which proves  \cref{eq:constraints-xu,eq:xr-in-X}. By applying the inference right before \cref{eq:xu-in-XUinterior-under-l1rg} again for $t\geq 0$, we obtain \cref{eq:xi-xni-bnd-under-l1rg}, which, together with $y_j(t)-y_{\nt,j}(t) = \sum_{i=1}^nC[j,i]\left(x_i(t)-x_{\nt,i}(t)\right)$, leads to \cref{eq:yi-yni-bnd-under-l1rg}.\end{proof}
\end{theorem}

\section{Simulation Results}\label{sec:sim}
We now apply \loneRG~to the longitudinal dynamics of an F-16 aircraft. The model was adapted from  \cite{sobel1985design-pitch} with slight modifications to remove the actuator dynamics, in which the state vector $x(t) = [\gamma(t), q(t), \alpha(t)]^\top$ consists of the flight path angle, pitch rate and angle of attack, and the control input vector $u(t)=[\delta_e(t), \delta_f(t)]$ includes the elevator deflection and flaperon deflection. The output vector is $y(t) = [\theta(t),\gamma(t)]^\top$, where
$\theta(t) = \gamma(t) + \alpha(t)$ is the pitch angle; the reference input vector is $r(t) = [\theta_c(t),\gamma_c(t)]^\top$, where $\theta_c$ and $\gamma_c$ are the commanded pitch
angle and flight path angle, respectively.
The system is subject to state and control constraints:
\begin{equation}\label{eq:cts-F16}
   \abs{\alpha(t)} \le 4 \!\textup{ deg},\   \abs{\delta_e(t)} \le 25 \!\textup{ deg}, \  \abs{\delta_f(t)} \le 22 \! \textup{ deg},  
\end{equation}
where the state constraint can also be represented as $x(t)\in\mcX\trieq [-10^3,10^3]\times [-10^3,10^3]\times[-4,4]$ following the convention in \cref{eq:constraints}. Furthermore, we assume 
\begin{equation}\label{eq:F16-init}
    \linfnorm{r}\leq 10, \quad x(0)\in \mcX_0 = \Omega(0.1).
\end{equation}
The open-loop dynamics are given by 
 { \setlength{\arraycolsep}{2pt}
\begin{align}
  \hspace{-1mm}  \dot x & \!=\! \begin{bmatrix}
    0 & 0.0067& 1.34 \\
    0 & -0.869 & 43.2 \\
    0 & 0.993 & -1.34
    \end{bmatrix} \!x+ 
    \begin{bmatrix}
    0.169 & 0.252 \\
    -17.3& -1.58 \\
    -0.169 & -0.252 
    \end{bmatrix}\! (u\!+\!f(t,x)), \label{eq:ol-dynamics-f16}
\end{align}}where $f(t,x) = [-0.8\sin(0.4\pi t)-0.1\alpha^2, 0.1-0.2\alpha]^\top$ is the uncertainty  dependent on both time and $\alpha$.
The feedback and feedforward gains of the baseline controller \cref{eq:baseline-control-law} are selected to be $K_x = [3.25,0.891,7.12;-6.10,-0.898,-10.0]$ and $K_v = [-3.93 ,0.679; 2.57, 3.53]$. Via simple calculations, we can see that $f(t,x) \in \mathcal W=[-2.4,2.4]\times [-0.9,0.9]$ when $x\in\mcX$ holds. 
\vspace{-4mm}
\subsection{\loneRG~Design}\label{sec:sub-l1rg-f16}
It can be verified that given any set $\mcZ$, $L_{f_1,\mcZ} = 0.2\max_{\alpha\in\mcZ_3}\abs{\alpha} $, $L_{f_2,\mcZ} =0.2$,  $b_{f_1,\mcZ} = 0.8+0.1\max_{\alpha\in\mcZ_3}\alpha^2$, $b_{f_2,\mcZ} =0.1+0.2\max_{\alpha\in\mcZ_3}\abs{\alpha}$ satisfy \cref{assump:lipschitz-bnd-fi}. 
For design of the \loneAC~in \cref{eq:l1-control-law,eq:adaptive_law,eq:state-predictor}, we select $A_e = -10I_3$ and parameterize the filter as $\mcC(s)= \frac{k_f}{s+k_f}I_2$, where $k_f>0$ denotes the bandwidth for both input channels. \cref{table:f16-l1-bnds} lists the bounds on $x_i(t)-x_{\nt,i}(t)$ and $u_j(t)-u_{\nt,j}(t)$ theoretically computed by applying \cref{alg:l1rg} under different $\mcC(s)$ and $T$ with and without using the scaling technique  in \cref{sec:l1ac-separate-bnds}. When applying the scaling technique, we set $\Tau_x^i[k]=0.01$ for each $i,k\in\mathbb{Z}_1^3$ and $k\neq i$, which satisfies \cref{eq:Tx-i-cts}.
Several observations can be made from~\cref{table:f16-l1-bnds}. First, by increasing the filter bandwidth $k_f$ and decreasing $T$, we are able to obtain a smaller $\gamma_1$ satisfying \cref{eq:T-constraint} and achieve tighter bounds for all states and inputs. In fact, if $k_f=10^3$ and $T=10^{-7}$, then $\tilde\rho_u^1$  $\tilde\rho_u^2$ are fairly close to the bounds on $f_1(t,x)$ and $f_2(t,x)$ for $x\in\mcX$, respectively,  which is consistent with \cref{rem:x-xid-discussion-Cs-T-refined}. 
Additionally, with scaling, we could significantly reduce $\tilde \rho^1$ and $\tilde \rho^3$, the bounds on $\gamma(t)-\gamma_\nt(t)$ and $\alpha(t)-\alpha_\nt(t)$, and $\tilde \rho_u^1$ and $\tilde\rho_u^2$, the bounds on $\delta_e(t)-\delta_{e,\nt}(t)$ and  $\delta_f(t)-\delta_{f,\nt}(t)$. Moreover, with $\Tau_x^3$, we can verify that the condition \cref{eq:l1-stability-condition-transformed} holds with $b_{f,\mcX_r}$ as long as $\linfnorm{v}< 1.868$. As mentioned in \cref{rem:conservatism-in-Cs-T}, the conditions \cref{eq:T-constraint,eq:l1-stability-condition-w-extra-Lf} and the resulting bounds $\tilde\rho^i$ and $\tilde \rho_u^j$ could be conservative. As a result, a larger reference command can potentially be allowed in a practical implementation while keeping $x(t)$ to stay in $\mcX$, as demonstrated in the following simulations.
{\setlength{\tabcolsep}{2pt}
\def\arraystretch{1.1}
\begin{table}[]
\centering
\caption{Performance bounds obtained under different filter bandwidth and sample time $T$ with and without (W/O)  scaling}\label{table:f16-l1-bnds}
\begin{tabular}{|l|cc|cc|}
\hline 
\multirow{2}{*}{\diagbox[width=1.5cm]{Var.}{Config.}}                     & \multicolumn{2}{c|}{$k_f=200,\ T=10^{-5}$} & \multicolumn{2}{c|}{$k_f=10^3,\ T=10^{-7}$} \\ \cline{2-5} 
                                              & \multicolumn{1}{l|}{W/O scaling}         & With scaling                 & \multicolumn{1}{l|}{W/O scaling}     & With scaling                    \\ \hline
$\gamma_1~/~b_{f,\mcX_r}$                                   & \multicolumn{2}{c|}{$0.01~/~2.40$}                                                 & \multicolumn{2}{c|}{$2\times 10^{-4}~/~2.40$}                                        \\ \hline
$[\tilde \rho^1,\tilde \rho^2, \tilde\rho^3]$ & \multicolumn{1}{l|}{$.41[1,1,1]$}           & $[.015,.41,.038]$         & \multicolumn{1}{l|}{$.043[1,1,1]$}      & $[.12,4.3,.35]10^{-2}$     \\ \hline
$[\tilde \rho_u^1,\tilde \rho_u^2]$           & \multicolumn{1}{l|}{$[8.15,9.02]$}           & $[4.20,2.85]$                & \multicolumn{1}{l|}{$[2.94,1.69]$}       & $[2.51, 1.03]$                     \\ \hline
\end{tabular}
\vspace{-5mm}
\end{table}}
Following \cref{alg:l1rg}, we used the bounds $\tilde \rho^3$, $\tilde \rho_u^1$ and $\tilde\rho_u^2$ obtained for the case when $k_f=200$ and $T=10^{-5}$,  to tighten the original constraints \cref{eq:cts-F16} and then used the tightened constraints to design the RG, for which we chose $T_d = 0.005$. Considering that $T_d$ was small, we did not consider inter-sample constraint violations and simply set $\hat \mcX_\nt = \mcX_\nt$ and $\hat \mcU_\nt = \mcU_\nt$ instead of \cref{eq:Xn-hat-Un-hat-defn}.
For comparisons, we also designed a robust RG (RRG) that treats the uncertainty $f(t,x)$ as a bounded disturbance $w(t)\in \mathcal W$, where $\mathcal W$ is introduced below \cref{eq:ol-dynamics-f16}. RRG design also uses $O_\infty$ set (defined in \cref{eq:Oinf-defn} for RG design); however, the prediction of the output, which corresponds to $\hat { \mathrm y}_c(k|v,x)$ for RG design, becomes a set-valued one taking into account all possible realizations of the disturbance $w(t)$ (see \cite{garone2017reference-governor-survey} for details). We additionally designed a standard RG by simply ignoring the uncertainty $f(t,x)$.
\vspace{-3mm}
\begin{figure}[h]
    \centering
    \includegraphics[width=0.6\columnwidth]{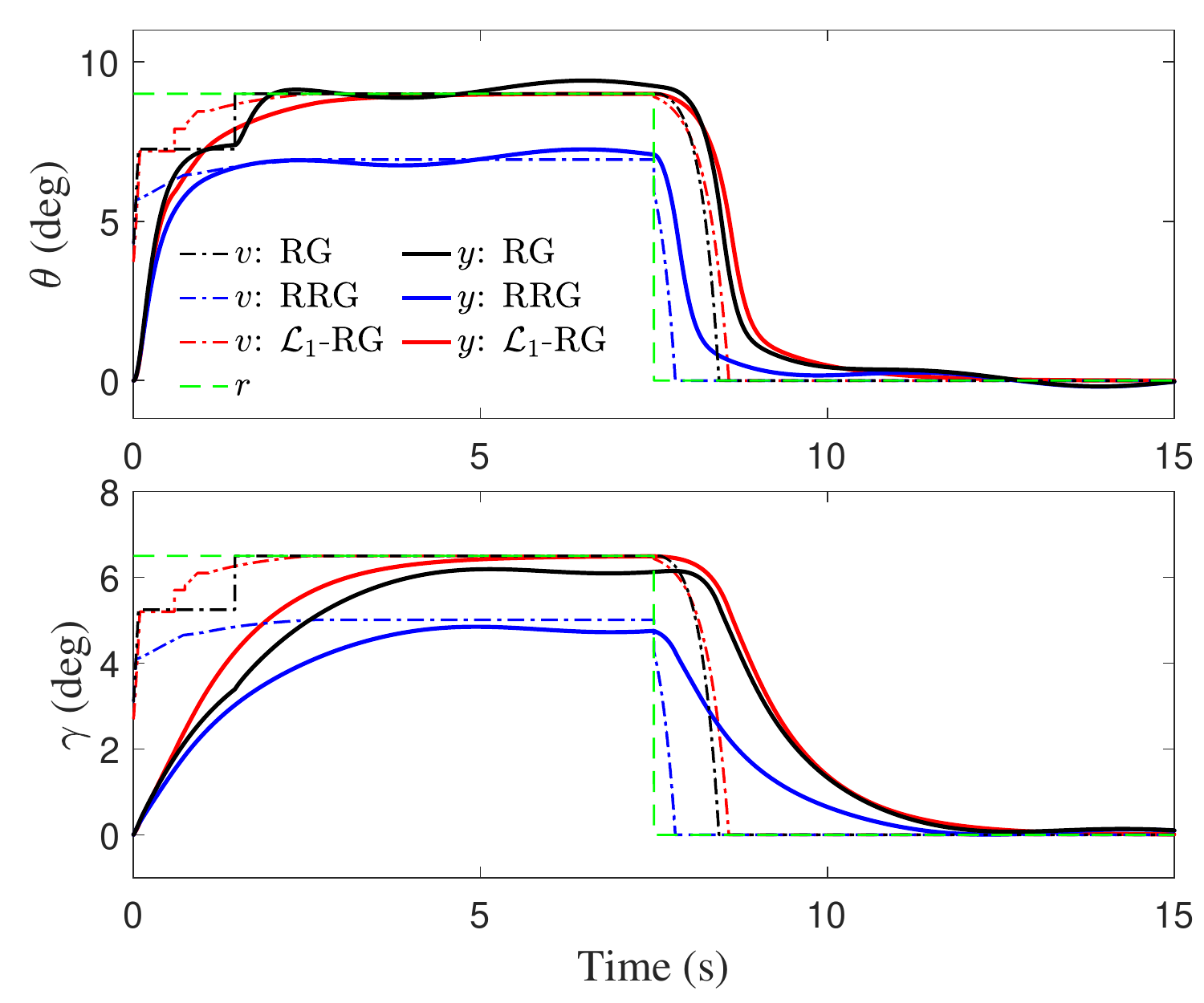} 
    \caption{Tracking performance under RG, RRG and \loneRG}\label{fig:tracking}
        \vspace{-3mm}
\end{figure}
\begin{figure}[h]
    \centering
    \includegraphics[width=0.6\columnwidth]{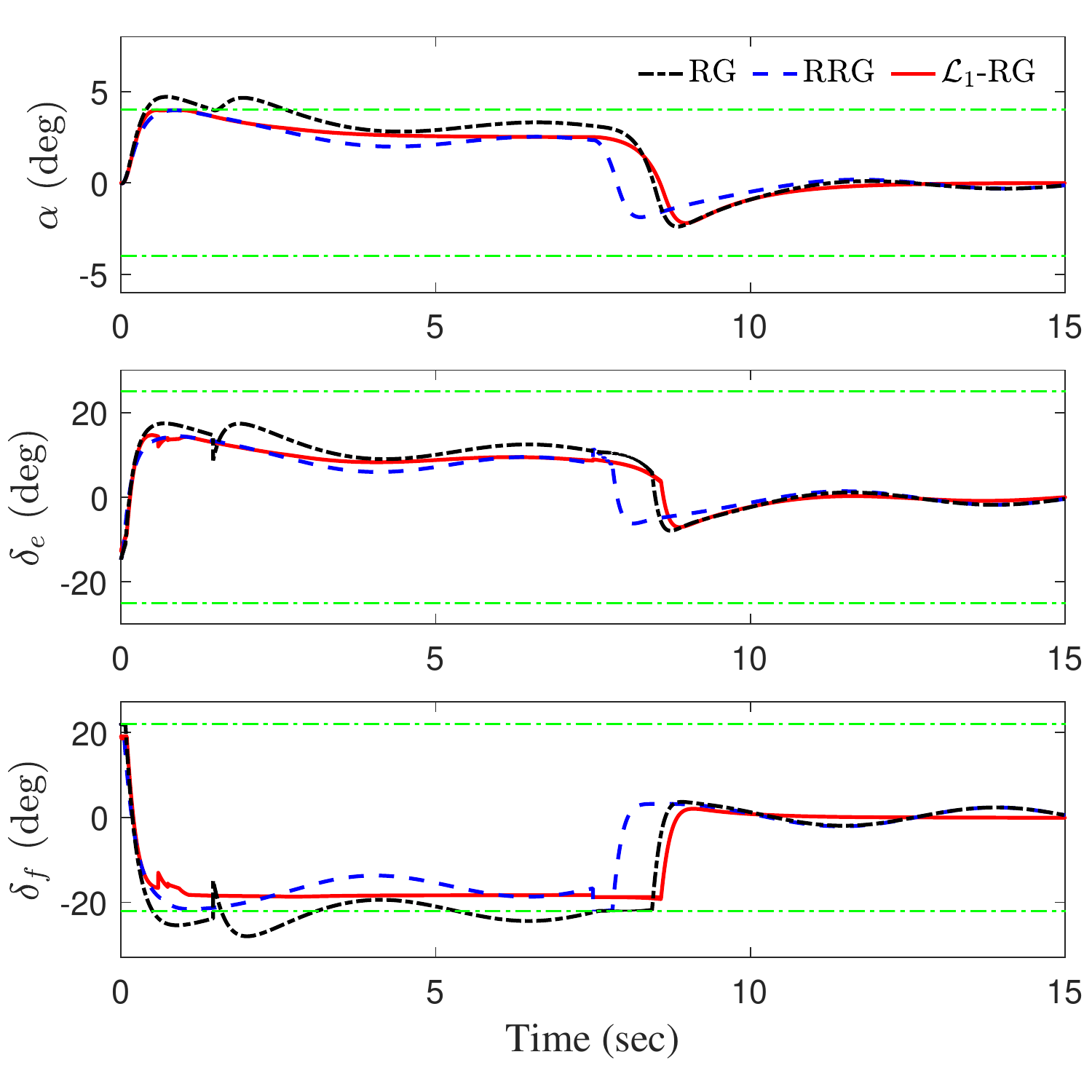} 
        \vspace{-3mm}
    \caption{Trajectories of constrained state and inputs under  RG, RRG and \loneRG. Green dash-dotted lines illustrate the constraints specified in \cref{eq:cts-F16}.}\label{fig:cts_vars}
        \vspace{-3mm}
\end{figure}
\begin{figure}[h]
    \centering
    \includegraphics[width=0.6\columnwidth]{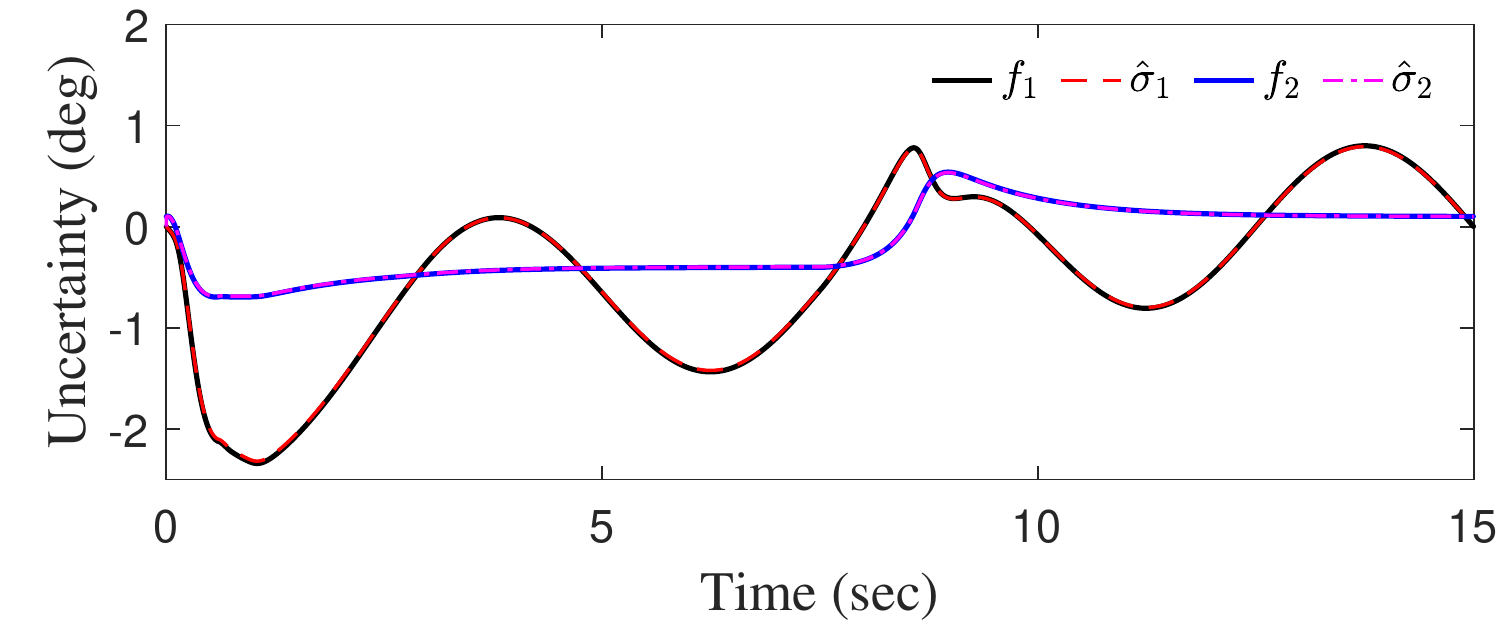} 
    \caption{Actual and estimated uncertainties under \loneRG. The symbol $f_j$ ($\hat\sigma_i$) denotes the $i$th element of $f$ ($\hat\sigma$), for $i=1,2$.}    \label{fig:uncert}
        \vspace{-3mm}
\end{figure}
\subsection{Simulation Results}
As mentioned in \cref{rem:conservatism-in-Cs-T}, the value of $T$ theoretically computed according to \cref{eq:T-constraint} is often unnecessarily small. For the subsequent simulations, we simply adopted an estimation sample time of 1 millisecond, i.e., $T=0.001$ s. As one can see in the subsequent simulation results, all the bounds  derived in \cref{sec:sub-l1rg-f16}  for $k_f=200$ and $T=10^{-5}$ still hold.
The reference command $r(t)$ was set to be $[9,6.5]$ deg for $t\in[0,7.5]$ s, and $[0,0]$ deg for $t\in[7.5,15]$ s. The results are shown in \cref{fig:uncert,fig:cts_vars,fig:tracking}.
 In terms of constraint enforcement, Fig.~\ref{fig:cts_vars} shows that both RRG and \loneRG~successfully enforced all the constraints, while violation of the constraints on the state $\alpha(t)$ and the input $\delta_f(t)$ happened under RG. However, from \cref{fig:tracking}, one can see that the RRG was quite conservative, leading to a large difference between the modified reference  and original reference commands and subsequently large tracking errors for both $\theta(t)$ and $\gamma(t)$ throughout the simulation. In comparison, the modified reference command under RG reached the original reference command, leading to better tracking performance.  Finally, \loneRG~yielded the best tracking performance, driving both $\theta(t)$ and $\gamma(t)$ very close to their commanded values at steady state. While noticeable under RG and RRG, the uncertainty-induced swaying in the outputs at steady state was negligible under \loneRG, thanks to the active compensation of the uncertainty by the \loneAC. From \cref{fig:uncert}, one can see that the estimation of the uncertainty within the \loneRG~was quite accurate. 
\begin{figure}[h]
    \centering
    \includegraphics[width=0.6\columnwidth]{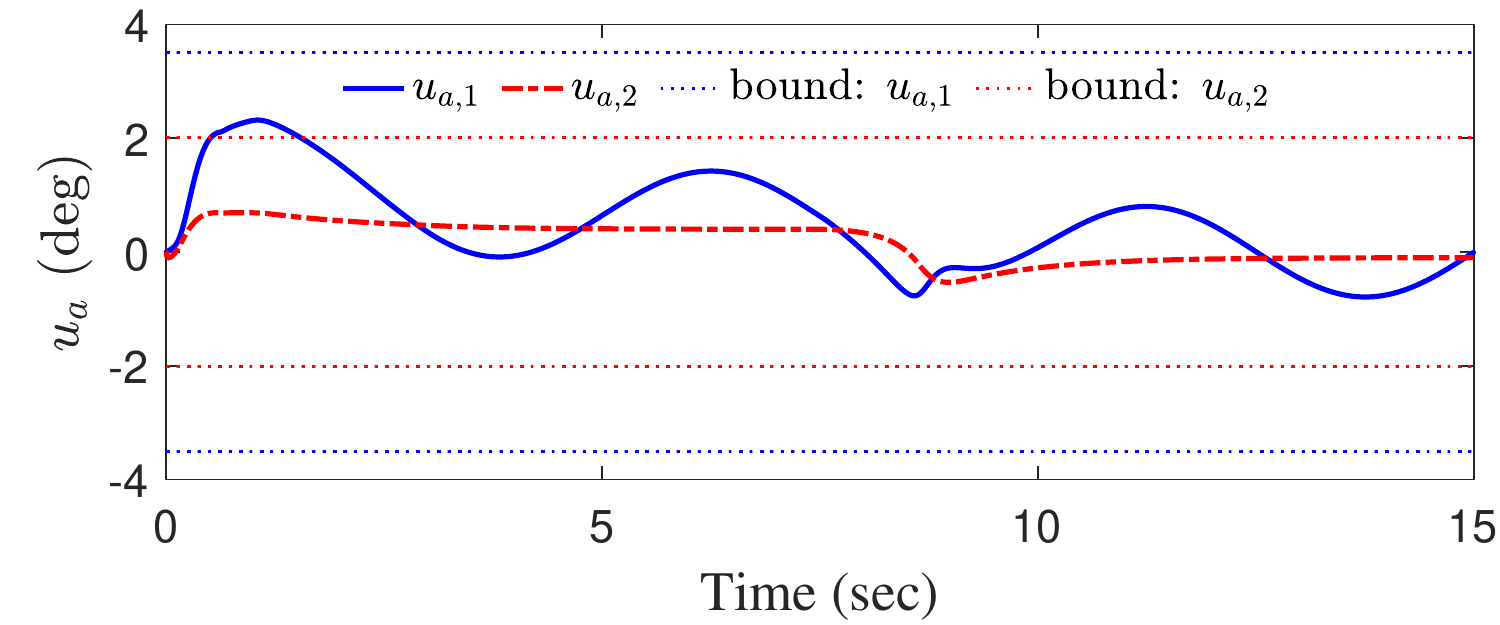} 
    \caption{Adaptive control inputs and theoretical bounds}
    \label{fig:l1inputs}
        \vspace{-3mm}
\end{figure}
\begin{figure}[h]
    \centering
    \includegraphics[width=0.6\columnwidth]{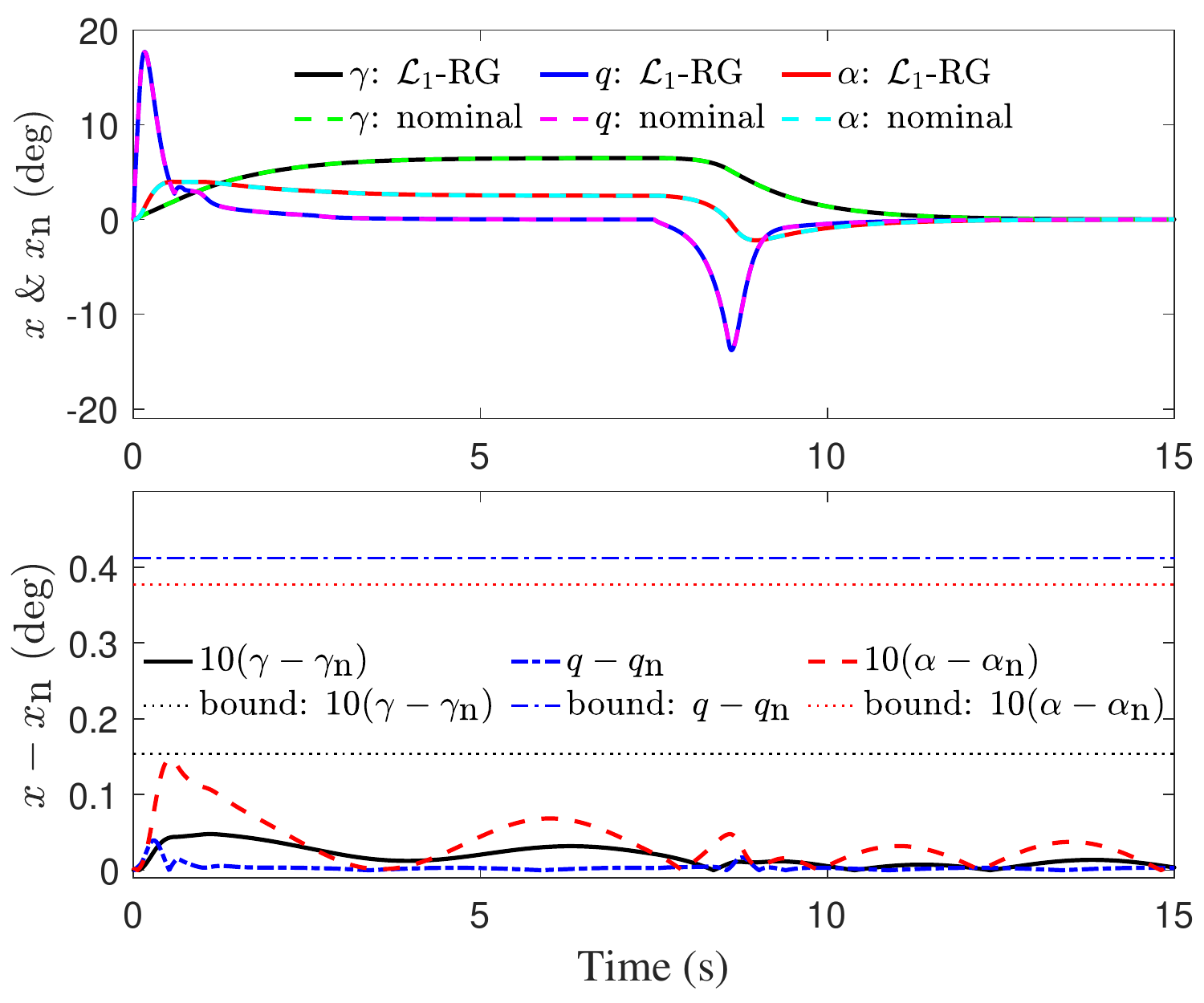} 
        \vspace{-2mm}
    \caption{Trajectories of states of the uncertain system ($x(t)$) under \loneRG~and of the nominal system ($x_\nt(t)$) under the same command $v(t)$ and their differences. The actual-nominal state errors and bounds for $\gamma(t)$ and $\alpha$ are scaled by 10 for a clear illustration.} 
    \label{fig:l1rg-nomrg}
    \vspace{-3mm}
\end{figure}
We next check whether the derived uniform bounds on the errors in states, $x(t)-x_\nt(t)$, and on the adaptive inputs, $u_\at(t)$, hold in the simulation. It can be seen from \cref{fig:l1inputs} that the bounds on both $u_{\at,1}(t)$ and $u_{\at,2}(t)$ were respected in the simulation and moreover are fairly tight. 
 Figure~\ref{fig:l1rg-nomrg} reveals that all actual states under \loneRG~were fairly close to their nominal counterparts, and moreover, the bound on $x_i(t)-x_{\nt,i}(t)$ for each $i\in \mbZ_1^3$ was respected. Note that $x_\nt(t)$ in \cref{fig:l1rg-nomrg} was produced by applying the same reference command $v(t)$ yielded by \loneRG~to the nominal system \cref{eq:nominal-cl-system-w-un}. 
\section{Conclusion}
In this paper, we developed \loneRG, an adaptive reference governor (RG) framework, for control of linear systems with time- and state-dependent uncertainties subject to both state and input constraints. At the core of \loneRG~is an \lonew adaptive controller that provides guaranteed uniform bounds on the errors between states and inputs of the uncertain system and those of a nominal (i.e., uncertainty-free) system. With such uniform error bounds for constraint tightening, a RG designed for the nominal system with tightened constraints guarantees the satisfaction of the original constraints by the actual states and inputs. Simulation results validate the efficacy and advantages of the proposed approach. 

In the future, we will address unmatched uncertainties following \cite{zhao2021RALPV}, and extend the proposed framework to the nonlinear setting leveraging the results in \cite{lakshmanan2020safe,zhao2022tube-rccm-ral}. Additionally, we would like to extend the proposed solution to adaptive MPC. 
\bibliographystyle{ieeetr}
\bibliography{bib/refs-pan,bib/refs-naira}

\begin{thebibliography}{10}

\bibitem{camacho2013mpc-book}
E.~F. Camacho and C.~B. Alba, {\em Model predictive control}.
\newblock Springer Science \& Business Media, 2013.

\bibitem{rawlings2020mpc-book}
J.~B. Rawlings, D.~Q. Mayne, and M.~M. Diehl, {\em {Model Predictive Control:
  Theory, Computation, and Design, 2nd Ed.}}
\newblock Nob Hill Publishing, 2020.

\bibitem{garone2017reference-governor-survey}
E.~Garone, S.~Di~Cairano, and I.~Kolmanovsky, ``Reference and command governors
  for systems with constraints: {A} survey on theory and applications,'' {\em
  Automatica}, vol.~75, pp.~306--328, 2017.

\bibitem{kolmanovsky1998robustRG-disturbance-invariant-sets}
I.~Kolmanovsky and E.~G. Gilbert, ``Theory and computation of disturbance
  invariant sets for discrete-time linear systems,'' {\em Mathematical Problems
  in Engineering}, vol.~4, no.~4, pp.~317--367, 1998.

\bibitem{kerrigan2001robust-mpc}
E.~C. Kerrigan, {\em Robust constraint satisfaction: {Invariant} sets and
  predictive control}.
\newblock PhD thesis, University of Cambridge, 2001.

\bibitem{langson2004robust-mpc-tube}
W.~Langson, I.~Chryssochoos, S.~Rakovi{\'c}, and D.~Q. Mayne, ``Robust model
  predictive control using tubes,'' {\em Automatica}, vol.~40, no.~1,
  pp.~125--133, 2004.

\bibitem{rakovic2005robust-mpc}
S.~Rakovic, {\em Robust control of constrained discrete time systems:
  {Characterization} and implementation}.
\newblock PhD thesis, University of London, 2005.

\bibitem{mayne2006robust-mpc}
D.~Q. Mayne, S.~V. Rakovi{\'c}, R.~Findeisen, and F.~Allg{\"o}wer, ``Robust
  output feedback model predictive control of constrained linear systems,''
  {\em Automatica}, vol.~42, no.~7, pp.~1217--1222, 2006.

\bibitem{mayne2011tube-mpc-nonlinear}
D.~Q. Mayne, E.~C. Kerrigan, E.~Van~Wyk, and P.~Falugi, ``Tube-based robust
  nonlinear model predictive control,'' {\em International Journal of Robust
  and Nonlinear Control}, vol.~21, no.~11, pp.~1341--1353, 2011.

\bibitem{kohler2020computationally-rmpc}
J.~K{\"o}hler, R.~Soloperto, M.~A. M{\"u}ller, and F.~Allg{\"o}wer, ``A
  computationally efficient robust model predictive control framework for
  uncertain nonlinear systems,'' {\em IEEE Transactions on Automatic Control},
  vol.~66, no.~2, pp.~794--801, 2020.

\bibitem{lopez2019dynamic-tube-mpc}
B.~T. Lopez, J.-J.~E. Slotine, and J.~P. How, ``Dynamic tube {MPC} for
  nonlinear systems,'' in {\em Proceedings of American Control Conference},
  pp.~1655--1662, 2019.

\bibitem{kouvaritakis2015mpc-classical-book}
B.~Kouvaritakis and M.~Cannon, {\em {Model Predictive Control: Classical,
  Robust and Stochastic}}.
\newblock Advanced Textbooks in Control and Signal Processing, Springer,
  London, 2015.

\bibitem{zhang2020adaptive-mpc-parametric}
K.~Zhang and Y.~Shi, ``Adaptive model predictive control for a class of
  constrained linear systems with parametric uncertainties,'' {\em Automatica},
  vol.~117, p.~108974, 2020.

\bibitem{adetola2009adaptive-mpc-nonlinear}
V.~Adetola, D.~DeHaan, and M.~Guay, ``Adaptive model predictive control for
  constrained nonlinear systems,'' {\em Systems \& Control Letters}, vol.~58,
  no.~5, pp.~320--326, 2009.

\bibitem{pereida2021robust-adaptive-mpc}
K.~Pereida, L.~Brunke, and A.~P. Schoellig, ``Robust adaptive model predictive
  control for guaranteed fast and accurate stabilization in the presence of
  model errors,'' {\em International Journal of Robust and Nonlinear Control},
  vol.~31, no.~18, pp.~8750--8784, 2021.

\bibitem{wang2017adaptiveMPC}
X.~Wang, L.~Yang, Y.~Sun, and K.~Deng, ``Adaptive model predictive control of
  nonlinear systems with state-dependent uncertainties,'' {\em Int. J. Robust
  Nonlinear Control}, vol.~27, no.~17, pp.~4138--4153, 2017.

\bibitem{bujarbaruah2020semi-adaptive-mpc}
M.~Bujarbaruah, S.~H. Nair, and F.~Borrelli, ``A semi-definite programming
  approach to robust adaptive {MPC} under state dependent uncertainty,'' in
  {\em European Control Conference}, pp.~960--965, IEEE, 2020.

\bibitem{naira2010l1book}
N.~Hovakimyan and C.~Cao, {\em $\mathcal{L}_1$ Adaptive Control Theory:
  Guaranteed Robustness with Fast Adaptation}.
\newblock Philadelphia, PA: Society for Industrial and Applied Mathematics,
  2010.

\bibitem{poloni2014disturbance-constraint-rg}
T.~Pol{\'o}ni, U.~Kalabi{\'c}, K.~McDonough, and I.~Kolmanovsky, ``Disturbance
  canceling control based on simple input observers with constraint enforcement
  for aerospace applications,'' in {\em IEEE Conference on Control
  Applications}, pp.~158--165, IEEE, 2014.

\bibitem{pin2009robust}
G.~Pin, D.~M. Raimondo, L.~Magni, and T.~Parisini, ``Robust model predictive
  control of nonlinear systems with bounded and state-dependent
  uncertainties,'' {\em IEEE Transactions on Automatic Control}, vol.~54,
  no.~7, pp.~1681--1687, 2009.

\bibitem{gilbert1991linear-constraints-maximal}
E.~G. Gilbert and K.~T. Tan, ``Linear systems with state and control
  constraints: {The} theory and application of maximal output admissible
  sets,'' {\em IEEE Transactions on Automatic control}, vol.~36, no.~9,
  pp.~1008--1020, 1991.

\bibitem{bemporad1995nonlinear-rg}
A.~Bemporad and E.~Mosca, ``Nonlinear predictive reference filtering for
  constrained tracking,'' in {\em Proceedings of European Control Conference},
  pp.~1720--1725, 1995.

\bibitem{gilbert1995discrete-rg}
E.~G. Gilbert, I.~Kolmanovsky, and K.~T. Tan, ``Discrete-time reference
  governors and the nonlinear control of systems with state and control
  constraints,'' {\em International Journal of Robust and Nonlinear control},
  vol.~5, no.~5, pp.~487--504, 1995.

\bibitem{Scherer97Multi}
C.~Scherer, P.~Gahinet, and M.~Chilali, ``Multiobjective output-feedback
  control via {LMI} optimization,'' {\em IEEE Transactions on Automatic
  Control}, vol.~42, no.~7, pp.~896--911, 1997.

\bibitem{zhao2021RALPV}
P.~Zhao, S.~Snyder, N.~Hovakimyana, and C.~Cao, ``Robust adaptive control of
  linear parameter-varying systems with unmatched uncertainties,'' {\em
  arXiv:2010.04600}, 2021.

\bibitem{sobel1985design-pitch}
K.~M. Sobel and E.~Y. Shapiro, ``A design methodology for pitch pointing flight
  control systems,'' {\em Journal of Guidance, Control, and Dynamics}, vol.~8,
  no.~2, pp.~181--187, 1985.

\bibitem{lakshmanan2020safe}
A.~Lakshmanan, A.~Gahlawat, and N.~Hovakimyan, ``Safe feedback motion planning:
  A contraction theory and $\mathcal{L}_1$-adaptive control based approach,''
  in {\em Proceedings of 59th IEEE Conference on Decision and Control (CDC)},
  pp.~1578--1583, 2020.

\bibitem{zhao2022tube-rccm-ral}
P.~Zhao, A.~Lakshmanan, K.~Ackerman, A.~Gahlawat, M.~Pavone, and N.~Hovakimyan,
  ``Tube-certified trajectory tracking for nonlinear systems with robust
  control contraction metrics,'' {\em IEEE Robotics and Automation Letters},
  pp.~1--1, 2022.

\end{thebibliography}

\appendix
\section{Proofs}
\subsection{Proof of \cref{lem:nominal-system-inter-sample-behavior}}\label{sec:sub-proof-for-lem:nominal-system-inter-sample-behavior}
\begin{proof}
Since the continuous-time system \cref{eq:nominal-cl-system-w-un} has the same states as the discrete-time system \cref{eq:nominal-system-discrete-w-un} at all sampling instants, if \cref{eq:xn-un-satisfy-tightened-csts} holds for \cref{eq:nominal-system-discrete-w-un}, then we have 
\begin{equation}\label{eq:xn-un-at-sampling-instants}
     \xn(kT_d)\in \hat \mcX_\nt\subset \mcX_\nt,\ \un(kT_d)\in \hat \mcU_\nt\subset\mcU_\nt,\quad \forall k\in\mbZ_+,
\end{equation}
for \cref{eq:nominal-cl-system-w-un}. Next we analyze the behavior of \cref{eq:nominal-cl-system-w-un} between adjacent sampling instants. Towards this end, consider any $t= k^\ast T_d+\tau$ for some $k^\ast\in\mbZ_+$ and $\tau\in[0,T_d)$. From \cref{eq:nominal-cl-system-w-un}{}, we have
$
  \xn(t) = \xn(k^\ast T_d+\tau) = ~e^{A_m\tau}\xn(k^\ast T_d)+\int_{k^\ast T_d}^{k^\ast T_d +\tau} e^{A_m(k^\ast T_d+\tau-\xi)}B_vv(\tau)d\xi 
     =  e^{A_m\tau}\xn(k^\ast T_d)+\int_{k^\ast T_d}^{k^\ast T_d +\tau} e^{A_m(k^\ast T_d+\tau-\xi)}d\xi B_vv(k^\ast T_d) = e^{A_m\tau}\xn(k^\ast T_d)+A_m^{-1}\left(e^{A_m\tau}-I_n\right) B_vv(k^\ast T_d), 
$
where the third equality is due to the fact that $v(k^\ast T_d+\tau) = v(k^\ast T_d)$ for all $\tau\in[0,T_d)$. As a result, we have $\xn(t)-\xn(k^\ast T_d) = \xn(k^\ast T_d+\tau)-\xn(k^\ast T_d) = \left(e^{A_m\tau}-I_n\right)\left(\xn(k^\ast T_d)+A_m^{-1}B_v v(k^\ast T_d) \right)$. Thus, we have 
\begin{subequations}\label{eq:xn-un-t-kTd-diff-bnd} 
\begin{align}
    \infnorm{\xn(t)-\xn(k^\ast T_d)} & \leq \nu(T_d) \label{eq:xn-t-kTd-diff-bnd}, \\
    \infnorm{\un(t)-\un(k^\ast T_d)} & \leq \infnorm{K_x}\nu(T_d) \label{eq:un-t-kTd-diff-bnd}, 
\end{align}
\end{subequations}
where $\nu(T_d)$ defined in \cref{eq:nu-Td-defn}, while \cref{eq:un-t-kTd-diff-bnd} is due to the fact that $\un(t)-\un(k^\ast T_d) = K_x \left( \xn(t)-\xn(k^\ast T_d) \right)$.
Considering \cref{eq:xn-un-t-kTd-diff-bnd,eq:xn-un-at-sampling-instants,eq:Xn-hat-Un-hat-defn}, we have $\xn(t)\in \mcX_\nt$ and $\un(t)\in \mcX_\nt$ for all $t\geq 0$. The proof is complete.
\end{proof}

\subsection{Proof of \cref{lem:ref-xr-ur-bnd}}\label{sec:sub-proof-for-lem:ref-xr-ur-bnd}
\begin{proof}
Rewriting the dynamics of the reference system in \eqref{eq:ref-system} in the Laplace domain yields
\begin{equation}\label{eq:xr-expression-w-G-f}
   \xr(s) = \mcG_{xm}(s)\laplace{f(t,x_\textup{r}(t))} + \mcH_{xv}(s)v(s)+\xin(s).
\end{equation}
Therefore, from \cref{lem:L1-Linf-relation}, for any $\xi>0$, we have 
\begin{equation}\label{eq:x_r_linfnorm_truc_bound}
    \linfnormtruc{x_\textup{r}}{\xi} \leq \lonenorm{\mcG_{xm}(s)} \linfnormtruc{\eta_\rt}{\xi}+ \lonenorm{\mcH_{xv}(s)}\linfnorm{v}+ \linfnorm{\xin},
\end{equation}
where $\eta_\rt(t)$ is defined in \cref{eq:eta-eta_r-defn}. 
If \cref{eq:xref-bnd} is not true, since $x_\textup{r}(t)$ is continuous and $\infnorm{x_\textup{r}(0)}<\rho_r$, there exists a $\tau\!>\!0$ such that 
\begin{equation}
    \infnorm{x_\textup{r}(t)}<\rho_r, \ \forall t\in[0,\tau),\ \textup{and}\ \infnorm{x_\textup{r}(\tau)}=\rho_r,
\end{equation} 
which implies 
$x_\rt(t)\in \Omega(\rho_r)$ for any $t$ in $[0,\tau]$. Further considering \cref{eq:bnd-f} that results from \cref{assump:lipschitz-bnd-fi}, we have
\begin{equation}\label{eq:f-xr-rhor-tau}
    \linfnormtruc{\eta_\rt}{\tau}\leq b_{f,\Omega(\rho_r)}.
\end{equation}
Plugging the preceding inequality into  \cref{eq:x_r_linfnorm_truc_bound} leads to 
\begin{equation}
    \rho_r \leq \lonenorm{\mcG_{xm}(s)} b_{f,\Omega(\rho_r)}+ \lonenorm{\mcH_{xv}(s)}\linfnorm{v} + \rho_\textup{in},
\end{equation}
which contradicts the condition \eqref{eq:l1-stability-condition}. Therefore, \eqref{eq:xref-bnd} is true. Equation \eqref{eq:uref-bnd} immediately follows from \eqref{eq:xref-bnd} and \cref{eq:ref-system}. 
\end{proof}

\subsection{Proof of \cref{lem:xtilde-bnd}}\label{sec:sub-proof-lem:xtilde-bnd}
\begin{proof}
Due to \cref{eq:x-u-tau-bnd-assump-in-lemma}, we have $x(t)\in \Omega(\rho)$ for any $t$ in $[0,\tau]$. Further considering \cref{eq:bnd-f} that results from \cref{assump:lipschitz-bnd-fi}, we have
\begin{equation}\label{eq:f-bnd-in-0-tau}
    \infnorm{f(t,x(t))} =\infnorm{\eta(t)} \leq b_{f,\Omega(\rho)},\quad \forall t\in [0,\tau].
\end{equation} 
From \eqref{eq:prediction-error}, for any $0\leq t <T$ and $i\in\mbZ_0$, we have
{
\begin{align}
    \tilx(iT+t) = & ~e^{A_et}\tilx(iT)+\int_{iT}^{iT+t} e^{A_e(iT+t-\xi)}[B \ B^\perp]
   \begin{bmatrix}
   \hsigma_1(iT) \\
   \hsigma_2(iT)
   \end{bmatrix}
   d\xi  - \int_{iT}^{iT+t} e^{A_e(iT+t-\xi)}B \eta(\xi) d\xi \nonumber \\ 
     = & ~e^{A_et}\tilx(iT)+\int_{0}^{t} e^{A_e(t-\xi)}[B \ B^\perp]
   \begin{bmatrix}
   \hsigma_1(iT) \\
   \hsigma_2(iT)
   \end{bmatrix}
   d\xi  - \int_{0}^{t} e^{A_e(t-\xi)}B \eta(iT+\xi) d\xi. 
    \label{eq:tilx-iTplust}
\end{align}}Considering the adaptive law \eqref{eq:adaptive_law}, the preceding equality implies
\begin{align}
  \hspace{-3mm}  \tilx((i+1)T) \!= \!- \!\int_{0}^{T}\! e^{A_e(T-\xi)}B \eta(iT\!+\!\xi) d\xi.
\end{align}
Therefore, for any $i\in \mbZ_0$ with $(i+1)T\leq \tau$, we have
\begin{align}
    \infnorm{\tilx((i+1)T)} &\leq \!\int_{0}^{T}\! \infnorm{e^{A_e(T-\xi)}B} \infnorm{\eta(iT\!+\!\xi)}d\xi \leq \bar \alpha_0(T) b_{f,\Omega(\rho)},
\end{align}
where $\bar \alpha_0(T)$ is defined in \cref{eq:alpha_0-defn}, and the last inequality is due to \cref{eq:f-bnd-in-0-tau}. 
Since $\tilx(0)=0$, we therefore have 
\begin{equation}\label{eq:tilx-iT-bnd}
    \infnorm{\tilx(iT)}\leq \bar \alpha_0(T) b_{f,\Omega(\rho)} \leq \gamma_0(T), \; \forall  iT\leq \tau, i\in \mbZ_0.
\end{equation}
Now consider any $t\in(0,T]$ such that $iT+t\leq \tau$ with $i\in \mbZ_0$. From \eqref{eq:tilx-iTplust} and the adaptive law \cref{eq:adaptive_law}, we have
\begin{align}
    \infnorm{\tilx(iT+t)} \leq  & \infnorm{e^{A_et}}\infnorm{\tilx(iT)} +  \int_{0}^{t} \infnorm{e^{A_e(t-\xi)}\Phi^{-1}(T)e^{A_eT}} \infnorm{\tilx(iT)}  d\xi   \nonumber \\
& +\int_0^t \infnorm{e^{A_e(t-\xi)}B} \infnorm{\eta(iT+\xi)} d\xi \nonumber \\ \leq &  \left(\bar \alpha_1(T)+\bar \alpha_2(T)+1 \right) \bar\alpha_0(T) b_{f,\Omega(\rho)} = \gamma_0(T),  \label{eq:tilx-iT+t-bnd} 
\end{align}
where $\bar\alpha_i(T)$ ($i=0,1,2$) are defined in \cref{eq:alpha_0-defn,eq:alpha_1-defn,eq:alpha_2-defn}, and the last inequality is partially due to the fact that $\int_0^t \infnorm{e^{A_e(t-\xi)}B}d\xi\leq\int_0^T \infnorm{e^{A_e(T-\xi)}B}\!d\xi\!=\!\bar\alpha_0(T) $. Equations \cref{eq:tilx-iT-bnd,eq:tilx-iT+t-bnd} imply \eqref{eq:tilx_tau-leq-gamma0}. 
\end{proof}

\subsection{Proof of \cref{them:x-xref-bnd}}\label{sec:sub-proof-them:x-xref-bnd}
\begin{proof}
We first prove \cref{eq:xref-x-bnd,eq:uref-u-bnd}  by contradiction. Assume \eqref{eq:xref-x-bnd} or \eqref{eq:uref-u-bnd} do not hold. Since $\infnorm{x_\textup{r}(0)-x(0)}=0<\gamma_1$ and $\infnorm{u_{\textup{r}}(0)-u_\at(0)}=0<\gamma_2$, and $x(t)$, $u_\at(t)$, $x_\textup{r}(t)$ and $u_{\textup{r}}(t)$ are all continuous, there must exist an instant $\tau$ such that
\begin{equation}
 \hspace{-2mm}   \infnorm{x_\textup{r}(\tau)-x(\tau)} = \gamma_1 \textup{ or }  \infnorm{u_{\textup{r}}(\tau)-u(\tau)} = \gamma_2,
\end{equation}
while 
\begin{equation} 
   \hspace{-3mm} \infnorm{x_\textup{r}(t)\!-\!x(t)}\! <\! \gamma_1, \  \infnorm{u_{\textup{r}}(t)\!-\!u(t)} \!<\! \gamma_2, \ \forall t\in [0,\tau).
\end{equation}
This implies that at least one of the following equalities hold:
\begin{equation}\label{eq:xr-x-ur-u-linfnorm-tau}
    \linfnormtruc{x_\textup{r}-x}{\tau} = \gamma_1, \quad \linfnormtruc{u_{\textup{r}}-u_\at}{\tau} = \gamma_2.
\end{equation}
Note that $\linfnorm{x_\textup{r}} \leq \rho_r<\rho $ according to \cref{lem:ref-xr-ur-bnd} and $\linfnorm{x}\leq \rho_r+\gamma_1 = \rho$ from \cref{eq:xr-x-ur-u-linfnorm-tau}. Further considering \cref{eq:lipschitz-cond-f} that results from \cref{assump:lipschitz-bnd-fi}, we have that
\begin{equation}\label{eq:f-xr-x-tau-bnd}
 \hspace{-1mm}   \infnorm{f(t,x_\textup{r}(t))\!-\!f(t,x(t))} \!\leq\! L_{f,\Omega(\rho)}\! \linfnormtruc{x_\textup{r}\!-\!x}{\tau}\!, \ \forall t \!\in \![0,\tau].
\end{equation}
The control laws in \cref{eq:l1-control-law} and  \cref{eq:ref-system} indicate 
\begin{align}
&   \hspace{-2mm} u_{\textup{r}}(s) - u_\at(s)  = -\mcC(s)\laplace{f(t,x_\textup{r})-\hsigma_1(s)}  \hspace{-2mm}  = \mcC(s)\laplace{f(t,x)\!-\!f(t,x_\textup{r})} + \mcC(s)(\hsigma_1(s) \!-\! \laplace{f(t,x)}).\label{eq:omega-ur-u}
\end{align}
Equation \eqref{eq:prediction-error} indicates that
\begin{equation}\label{eq:hsigma-sigma-s}
    \hsigma_1(s) - \laplace{f(t,x)}) = B^\dagger(sI_n-A_e)\tilx(s).
\end{equation}
Considering \eqref{eq:dynamics-uncertain}, \cref{eq:l1-control-law} and \cref{eq:hsigma-sigma-s}, we have 
\begin{equation}
     x(s) = \mcG_{xm}(s) \laplace{f(t,x)} + \mcH_{xv}(s)v(s) + \xin(s)   - \mcH_{xm}(s)\mcC(s)B^\dagger(sI_n-A_e)\tilx(s),
\end{equation}
  which, together with \cref{eq:xr-expression-w-G-f}, implies
  \begin{equation}
x_\textup{r}(s) - x(s) = \mcG_{xm}(s) \laplace{f(t,x_\textup{r})-f(t,x)}  +\mcH_{xm}(s)\mcC(s)B^\dagger(sI_n-A_e)\tilx(s).
\end{equation}
Therefore, further considering \eqref{eq:f-xr-x-tau-bnd} and \cref{lem:xtilde-bnd}, we have 
\begin{align}
    \linfnormtruc{x_\textup{r}-x}{\tau} \leq & \lonenorm{\mcG_{xm}} L_{f,\Omega(\rho)} \linfnormtruc{x_\textup{r}-x}{\tau}  +\! \lonenorm{\mcH_{xm}(s)\mcC(s)B^\dagger(sI_n-A_e)}\!\gamma_0(T).\nonumber
\end{align}
The preceding equation, together with \cref{eq:l1-stability-condition-Lf}, leads to 
\begin{align}
   \hspace{-2mm} \linfnormtruc{x_\textup{r}-x}{\tau} &\!\leq\! \frac{\lonenorm{\mcH_{xm}(s)\mcC(s)B^\dagger(sI_n\!-\!A_e)}}{1-    \lonenorm{\mcG_{xm}}  L_{f,\Omega(\rho)}} \gamma_0(T),
\end{align}
which, together with the sample time constraint \cref{eq:T-constraint}, indicates that 
\begin{equation}\label{eq:xr-x<gamma1}
    \linfnormtruc{x_\textup{r}-x}{\tau}  < \gamma_1. 
\end{equation}

On the other hand, it follows from \cref{eq:f-xr-x-tau-bnd,eq:omega-ur-u,eq:hsigma-sigma-s,eq:xr-x<gamma1} that
\begin{align*}
  \linfnormtruc{ u_{\textup{r}}- u_\at}{\tau}  &  \leq  \lonenorm{\mcC(s)} L_{f,\Omega(\rho)}  \linfnormtruc{x_\textup{r}-x}{\tau} + \lonenorm{\mcC(s)B^\dagger(sI_n-A_e)}\linfnormtruc{\tilx}{\tau} \nonumber \\
    & < \lonenorm{\mcC(s)} L_{f,\Omega(\rho)} \gamma_1 + \lonenorm{\mcC(s)B^\dagger(sI_n-A_e)}\gamma_0(T).
\end{align*}
Further considering the definition in \cref{eq:gamma2-defn}, we have
\begin{equation}\label{eq:ur-u<gamma2}
    \linfnormtruc{ u_{\textup{r}}- u_\at}{\tau}<\gamma_2.
\end{equation}
Note that \cref{eq:xr-x<gamma1} and \cref{eq:ur-u<gamma2} contradict the equalities in \cref{eq:xr-x-ur-u-linfnorm-tau}, which proves \cref{eq:xref-x-bnd,eq:uref-u-bnd}. The bounds in \cref{eq:x-bnd,eq:ua-bnd} follow directly from  \cref{eq:xref-x-bnd,eq:uref-u-bnd,eq:xref-bnd,eq:uref-bnd} and the definitions of $\rho$ and $\rho_{u_\at}$ in \cref{eq:rho-defn,eq:rho-u-defn}.
The proof is complete. 
\end{proof}

\subsection{Proof of \cref{lem:ref-id-bnd}}\label{sec:sub-proof-lem:ref-id-bound}
\begin{proof}
From \cref{eq:nominal-cl-system,eq:ref-system}, we have 
\begin{equation}\label{eq:xr-xn-expression}
\hspace{-2mm}   \xr(s)-\xn(s) = G_{xm}(s) \laplace{f(t,x_\textup{r})} = G_{xm}(s) \laplace{\eta_\textup{r}(t)}.
\end{equation}
According to \cref{lem:ref-xr-ur-bnd},  we have $x_\rt(t) \in \Omega(\rho_r)$ for any $t\geq0$. Further considering 
 \cref{eq:bnd-f} that results from \cref{assump:lipschitz-bnd-fi}, we have $\linfnorm{\eta_\rt}\leq b_{f,\Omega(\rho_r)} $, which, together with \cref{eq:xr-xn-expression}, leads to \cref{eq:xref-xid-bnd}. 
\end{proof}

\end{document}